\documentclass[graybox]{svmult}
\usepackage{mathptmx}       
\usepackage{helvet}         
\usepackage{courier}        
\usepackage{type1cm}        
\usepackage{makeidx}         
\usepackage{graphicx}        
\usepackage{multicol}        
\usepackage[bottom]{footmisc}

\usepackage{amssymb,amsmath,latexsym,mathrsfs}
\usepackage{tikz}
\usetikzlibrary{decorations.pathmorphing,positioning,shapes,arrows}
\usetikzlibrary{arrows,positioning,fit,calc,decorations.pathmorphing}
\tikzset{every picture/.style={>=latex,semithick}}

\usepackage{algorithm}
\usepackage{algpseudocode}
\usepackage{array}
\usepackage{blkarray}
\usepackage{cite}
\usepackage{color}
\usepackage{comment}
\usepackage{enumerate}
\usepackage{epstopdf}
\usepackage{float}
\usepackage{gastex}
\usepackage{hyperref}
\usepackage{latexsym}
\usepackage{longtable}
\usepackage{multirow}
\usepackage{paralist}
\usepackage{pgf}
\usepackage{rotating}
\usepackage{siunitx}
\usepackage{subfigure}
\usepackage{thmtools}
\usepackage{times}
\usepackage{todonotes}
\usepackage{xcolor}
\usepackage{xspace}
\usepackage{url}

\hypersetup{colorlinks=false,pdfborder={0 0 0},breaklinks=true}
\renewcommand{\l}{\ell}

\algrenewcommand\alglinenumber[1]{\tiny #1:}

\newcommand{\mc}{M}

\newcommand{\move}[1]{\stackrel{#1}{\rightarrow}}

\renewcommand{\node}{\rho}

\newcommand{\Prob}{P}

\def\bx{{\textbf{x}}}

\def\by{{\textbf{y}}}

\def\bv{{\textbf{v}}}

\declaretheorem[name=Theorem]{thm}

\newcommand{\U}{\mathbf{U}}

\newcommand{\citep}{\cite}

\newcommand{\Ut}{\mathbf{U}^{t}}
\newcommand{\Uk}{\mathbf{U}^{k}}
\newcommand{\Ult}{\mathbf{U}^{\leq t}}

\newcommand{\Ulk}{\mathbf{U}^{\leq k}}
\newcommand{\Glt}{\mathbf{G}^{\le t}}

\newcommand{\Flt}{\mathbf{F}^{\le t}}
\newcommand{\Ft}{\mathbf{F}^{t}}

\renewcommand{\Re}{\mathbb{R}}

\newcommand{\vv}{\textbf{v}}
\newcommand{\vx}{\textbf{x}}

\newcommand{\ww}{\textbf{w}}
\newcommand{\vV}{\textbf{V}}
\newcommand{\vW}{\textbf{W}}
\newcommand{\wW}{\textbf{W}}

\newcommand{\ra}{\rightarrow}
\newcommand{\tm}{\mathbb{T}}

\renewcommand{\l}{\ell}
\newcommand{\tildex}{\raise.17ex\hbox{$\scriptstyle\mathtt{\sim}$}}

\def\bx{{\textbf{x}}}

\def\by{{\textbf{y}}}

\def\bv{{\textbf{v}}}

\def\ra{\rightarrow}

\newcounter{theass} 

\newcounter{theeg} 

\newcommand{\INIT}{\textup{INIT}}
\newcommand{\bP}{\textbf{P}}

\newcommand{\paths}{\mathsf{paths}}

\let\subparagraph\paragraph

\makeindex             

\begin{document}

\title*{Statistical Model Checking based Analysis of Biological Networks}
\titlerunning{Statistical Model Checking based Analysis of Biological Networks}
\author{Bing Liu, Benjamin M. Gyori and P.S. Thiagarajan}
\authorrunning{B. Liu, B. M. Gyori, and P.S. Thiagarajan} 
\institute{Bing Liu \at Department of Computational and Systems Biology, University of Pittsburgh, Pittsburgh, PA 15237, U.S.A., \email{liubing@pitt.edu}
\and Benjamin M. Gyori \at Laboratory of Systems Pharmacology, Harvard Medical School, Boston, MA 02115, U.S.A., \email{benjamin_gyori@hms.harvard.edu}
\and P.S. Thiagarajan \at Laboratory of Systems Pharmacology, Harvard Medical School, Boston, MA 02115, U.S.A., \email{thiagu@hms.harvard.edu}}
\maketitle

\abstract{We introduce a framework for analyzing ordinary differential equation (ODE) models of biological networks using statistical model checking (SMC). A key aspect of our work is the modeling of single-cell variability by assigning a probability distribution to intervals of initial concentration values and kinetic rate constants. We propagate this distribution through the system dynamics to obtain a distribution over the set of trajectories of the ODEs. This in turn opens the door for performing statistical analysis of the ODE system's behavior. To illustrate this we first encode quantitative data and qualitative trends as bounded linear time temporal logic (BLTL) formulas. Based on this we construct  a parameter estimation method using an SMC-driven evaluation procedure applied to the stochastic version of the behavior of the ODE system. We then describe how this SMC framework can be generalized to hybrid automata by exploiting the given distribution over the initial states and the--much more sophisticated--system dynamics to associate a Markov chain with the hybrid automaton. We then establish a strong relationship between the behaviors of the hybrid automaton and its associated Markov chain. Consequently, we sample trajectories from the hybrid automaton in a way that mimics the sampling of the trajectories of the Markov chain. This enables us to verify approximately that the Markov chain meets a BLTL specification with high probability. We have applied these methods to ODE based models of Toll-like receptor signaling and the crosstalk between autophagy and apoptosis, as well as to systems exhibiting hybrid dynamics including the circadian clock pathway and cardiac cell physiology. We present an overview of these applications and summarize the main empirical results. These case studies demonstrate that the our methods can be applied in a variety of practical settings.}

\vspace{-0.5cm}
\section{Introduction}
\label{sec:introduction}
\vspace{-0.2cm}
Many fundamental cellular functions are governed by biochemical networks--often called biopathways--via molecular interactions. Major diseases \citep{ferrari00} are caused by their malfunctioning. Hence it is vital to understand their behaviors. Mathematical models have the potential to elucidate the complex behavior of biopathways in a principled way \citep{aldridge_physicochemical_2006}. However it is very challenging to construct and analyze such models. Here we address the analysis  challenge by presenting methods for parameterizing and analyzing a rich class of models with respect to qualitative and quantitative constraints, in a setting with variable behaviors across a cell population.

Broadly speaking, the sources of variability in a population of cells can be due to both  intrinsic and extrinsic components.
Intrinsic variability is often due to stochasticity at the level of biochemical reactions or gene expressions which, when evolving concurrently can lead to distinct cellular states across a population. Extrinsic variability has at least two major sources. First, differences in the abundances of proteins present in each cell initially (corresponding to initial conditions in a model) are a major source of variability~\citep{spencer2009non}. Second, due to variability in environmental conditions encountered by each cell, the kinetic rates of reactions (corresponding to model parameters) will also vary across cells~\cite{snijder2011origins,weisse2010quantifying}. Intrinsic variability is often handled by using a probabilistic model such as the continuous-time Markov chains (CTMCs) \cite{wilkinson2011stochastic}. However, it is computationally difficult to calibrate and analyze such models--especially for large pathways--using stochastic simulation, due to the fact that they track the exact number of molecular species present in the network as well as individual reaction occurrences. On the other hand, deterministic models such as ordinary differential equations (ODEs) cannot naturally handle stochasticity at the level of individual reactions but they are computationally more tractable. These models are adequate when the number of molecules of each type involved in the pathway is high \cite{klipp2005systems}.

In this chapter we discuss methods to analyze ODE based dynamical models while accounting for such variability \cite{palaniappan2013,gyori2015}. First, we introduce a method by which the deterministic ODE dynamics is transformed into a stochastic one. The starting point for this method is to assume a probability distribution over an interval of values for each kinetic parameter and initial state. One can then use our method to quantify how specific dynamical properties vary as a result of the assumed distribution over the initial value states and parameters. We then discuss the extension of these methods to hybrid systems in which discrete switching between modes of ODE-governed continuous dynamics are modeled. In both cases, we build on statistical model checking (SMC) as the basis for model calibration and analysis. Finally, we give a combined overview of how these approaches have been applied to various biological systems and corresponding models \cite{Liu2016,Liu2017}.

Our approach to the probabilistic analysis of ODEs starts by requiring that the vector field associated with the system of ODEs is continuously differentiable. Based on this, we show that a distribution over initial states of the system induces a distribution over the space of finite trajectories of the system. For our intended applications this continuity requirement is easily met. The induced distribution over the finite trajectories at once leads to a simple SMC technique which we use to analyze the ODE system. An important strength of SMC is that its computational complexity is independent of the system size, and it therefore scales well. Further, formulating the analysis problem as a hypothesis test evaluated sequentially allows early stopping which leads to a reduced number of samples required to make a decision \cite{younes_numerical_2006}.

To illustrate the relevance of our approach, we develop a parameter estimation method. The set of unknown parameters in ODE models generally consist of kinetic rate constants and initial concentrations. Here, we assume for convenience that the nominal values of the initial conditions are known but that they fluctuate around this value in a bounded manner over the cell population. 

As might be expected, the parameter estimation procedure aims to find a set of parameters that best fits the given data, and at the same time can predict new behaviors~\cite{moles_parameter_2003}. To be able to take into account both quantitative time course experimental data and qualitative dynamical behaviors in parameter estimation, we represent both types of information as bounded linear time temporal logic (BLTL) formulas. For a given set of parameters, we compute its objective value using our SMC approach which assesses how well the model matches the data with the parameters (while taking into account the fluctuation of these values across the cell population giving rise to the data). We then employ a SERS \cite{sres} based algorithm to search in the parameter space for the optimal set of parameters with the smallest SMC based objective value. One can also construct a sensitivity analysis procedure by following a similar strategy \cite{palaniappan2013}.

We then describe an extension  of our SMC based analysis to the setting of multi-mode dynamics \cite{gyori2015}. One is 
interested in  this richer class of dynamics since many biological processes have multiple modes of operation 
where different portions of the underlying biochemical network switch in and out of action in different modes. 
The cell cycle, the operation of cardiac cells and the circadian clock pathway \cite{thecell,fenton08,miyano06} are  representative examples of such systems.
We use the standard formalism of hybrid automata  to capture multi-mode dynamics. As usual at each mode the system variables evolve continuously according to an associated system of ODEs. When the system enters a designated portion of the state space described by guards associated with the mode transitions, the system moves to a new mode instantaneously where it starts to evolve according to the ODEs associated with the new mode. 

Due to their highly expressive dynamics, hybrid automata are difficult to analyze effectively \cite{henzinger-lics-survey}. A common approach to achieve tractability is to restrict the mode dynamics in various ways \cite{girard2006efficient,frehse2005phaver,clarke2003verification,alur2000discrete,agrawal2006behavioural,henzinger1999discrete}. However, restricting the dynamics is not a good option for systems biology models, since the dynamics being studied is almost always non-linear. Accordingly, we lift our SMC based analysis technique to hybrid automata without unduly restricting their dynamics.

As before, we assume that the values of initial states follow a probability distribution. To simplify the presentation, the values of rate constants are assumed to be fixed. A major complication in dealing with the dynamics is that for a given trajectory, the value states and time points at which a guard is satisfied--and thus a mode change can occur--will be determined by the solution of the ODEs defined in each mode. Unfortunately, these solutions typically cannot be constructed in closed-form for high-dimensional systems. We deal with this by approximating the transition between two modes as a stochastic event. More precisely, we define the transition probability to be proportional to the measure on the set of \textit{(value state, time point)} pairs at which the guard corresponding to the transition is satisfied. To ensure a sound mathematical basis we further assume: (i) for each mode, the vector field associated with the mode's ODE system is Lipschitz continuous. (ii) The hybrid system's states are observed only at discrete time points. (iii) The guard sets and the set of initial states are bounded open sets. 
(iv) There is an upper bound on the number of transitions that can be triggered in a unit time interval (i.e. the system is strictly ``non-Zeno''). In fact, for technical convenience, we will assume that the time discretization has been chosen such that no more than one mode transition occurs between two successive time points. As might be expected, our procedure for sampling trajectories will crucially depend on the discrete time instances and will fail to detect multiple mode changes that may take place within a unit time interval. Thus there is a tradeoff between accuracy and efficiency involved in the choice of the discrete time interval. However in many biological settings successive mode changes within very small time intervals are unlikely and hence imposing the strong version of ``non-Zeno'' mode changes is realistic.
 
In this setting we then establish that the behavior of the hybrid system $H$ can be approximated as an infinite state Markov chain $M$ whose state space will be ``tree like''. For our purposes it will suffice to focus on the hybrid system's behavior within a bounded time horizon. Consequently we fix BLTL~\cite{clarke1999model} to be the specification logic. A key feature of our approximation is that $H$ satisfies the specification $\psi$--i.e. every trajectory of $H$ is a model of $\psi$--iff $M$ satisfies the specification $\psi$ with probability $1$. We can therefore use the Markov chain representation to approximately verify properties of interest for the hybrid automaton. However, $M$ can only be defined mathematically since constructing even a bounded initial portion of it can be intractable. This is due to the fact that the transition probabilities of $M$ between modes depends on the ODEs for which there is no closed-form solution available. For the same reason, whether a transition exists between a pair of states also can not be determined effectively. With this in mind, we design a SMC procedure to check that the chain almost surely satisfies a specification (implying satisfaction in the hybrid system as well) by generating samples from the Markov chain. To this end, we sample random trajectories of $H$ using numerical simulation such that it mimics sampling paths from the associated Markov chain. Note that a naive Monte Carlo simulation based strategy in which one samples a starting values state from the initial distribution and then generates a trajectory will have difficulty in ``randomly'' picking a mode transition. Instead, our approximation approach samples trajectories from $H$ in a principled way that mirrors the sampling of a trajectory from $M$.

We next present the results of applying our analysis methods to both single system of ODEs and hybrid systems. We have elsewhere \cite{palaniappan2013} SMC based parameter estimation and analysis methods to multiple pathway models available in the BioModels database \cite{biomodel}. Here we present the results for the myosin light chain (MLC) phosphorylation pathway, a large model governed by $105$ ODEs. $100$ of the model's $197$ rate constants are assumed to be unknown. We applied our SMC based parameter estimation technique. The results presented in Sect. \ref{sec:results} show that our approach is efficient and scales well. We also describe key results of the subsequent applications of this approach--in collaboration with experimentalists--to Toll-like receptor pathways \cite{Liu2016} and a model of autophagy-apoptosis crosstalk \cite{Liu2017}.

We have also validated our SMC based analysis method for hybrid automata using a number of case studies \cite{gyori2015}. Here we present the results for a hybrid system model of the cardiac heart cell \cite{fenton08}. Using this model and our SMC based analysis method we explore dynamical properties of importance on a multitude of cell types under a variety of input stimuli, and in healthy as well as diseased conditions.

As our results show the SMC based analysis framework for a single ODE system scale well. In the case of hybrid automata the cardiac cell model studied here involves only $4$ continuous variables while the circadian rhythm model analyzed in \cite{gyori2015} has $11$ continuous variables. However this has more to do with models we have studied rather than an inherent limitation on the dimension (the number of variables and modes) of the hybrid systems that can be analyzed using our method.

\vspace{-0.8cm}
\subsection{Related work}
\vspace{-0.3cm}
There are considerable previous studies on model checking of stochastic models \citep{jha_bayesian_2009,Heath2008,LiChen_MolecularBiosystem_2011,liu_approximate_2012}. Relevant to the present context is the research reported  in \cite{donaldson2008monte,Donaldson_2008_article} where probabilistic properties are verified by sampling a fixed number of trajectories from the model's state space. Hence this work does not provide any statistical guarantees. Instead, SMC based methods reported in \cite{Clarke_2008,jha_bayesian_2009} adaptively generate sampled trajectories to verity probabilistic properties and provide a statistical guarantee that can be specified by the user. 

Turning to model calibration, \cite{LiChen_MolecularBiosystem_2011} employed a brute force search with temporal logic constraints to estimate the parameters of Petri net models. For ODE models, parameter estimation combined with model checking has been carried out using different search strategies, including brute force search \cite{calzone2006machine}, genetic algorithms \cite{Donaldson_2008_article}, and covariance matrix adaptation evolution strategy \cite{rizk_continuous_2008}. However, the limitation of these studies is that each candidate set of parameters will be assessed by either a single trajectory \cite{calzone2006machine,rizk_continuous_2008} or only a fixed number of samples \cite{Donaldson_2008_article}.
Further, these studies did not validate the quality of the parameter estimates using am independent test data set. Turning to sensitivity analysis, \cite{Donaldson_2008_article} proposed to quantify the significance of a parameter with respect to a property by counting the number of parameter values using which the system meets the specification. Our SMC framework enables a principled and systematic property-based sensitivity analysis as described in \cite{palaniappan2013}.

As for related work on hybrid automata, \cite{abate2005stochastic} uses probabilistic barrier functions to substitute the guards in order to approximate mode transitions as random events. In our setting, we constructed the transition probabilities in a similar but simpler manner. This has enabled us to verify properties specified by temporal logic using sampling and numerical simulations. On the other hand, $\delta$-reals \cite{deltareach} adopts a different approximation approach, which aims to tackle the reachability problem by allowing small perturbations. Finally, stochastic hybrid automata and their analysis have been extensively studied~\cite{cassandras2010stochastic,blom2006stochastic,julius2009approximation,ballarini2011cosmos}.

\vspace{-0.8cm}
\subsection{Outline of the chapter}
\vspace{-0.3cm}
In the next section, we give an introduction to ODE models and their trajectories while incorporating intervals of initial states for the variables and the rate constants. In Sect. \ref{sec:smc-ode} we present BLTL and its semantics followed by the SMC procedure for ODE systems. We then describe how this can be used as the basis for performing parameter estimation. From here we lift the key ideas to the setting of hybrid automata. To do this, in Sect. \ref{sec:hybrid-model} we give an introduction to hybrid automata and their dynamics. This is followed by the mathematical construction of their Markov chain approximation. We then develop the SMC procedure for the Markov chain approximation and an algorithm for sampling trajectories from the hybrid automaton that mimics the sampling of the associated Markov chain. We present a number of applications and case studies in Sect \ref{sec:results}. We summarize the key aspects of this chapter in the final section. 

\vspace{-0.5cm}
\section{Pathway models based on a system of ODEs}
\vspace{-0.3cm}
\label{sec:model}
Here we introduce ODE based models of biopathways. Our models will account for variability in the initial conditions and the rate constants. Much of material in this section is abbreviated from \cite{palaniappan2013} and is used here for fixing the technical background.

\vspace{-0.8cm}
\subsection{ODEs preliminaries}
\vspace{-0.3cm}
\label{sec:ode}
Our ODE system involves a set of variables $\{x_1, x_2, \ldots, x_n\}$. They will correspond to the molecular species in the pathway 
while $\{\theta_1, \theta_2, \ldots, \theta_m\}$ will constitute the set of rate constants (parameters) in the ODE system. There will be one ODE for each  $x_i$ and it  will be of the form: 
\begin{equation*}
\frac{dx_i}{dt} = f_i(\textbf{x}_i,\Theta_i).
\end{equation*} 
Here  $f_i$ describes the dynamics of the reactions $R_i$  in which $x_i$ participates (as a reactant or product). $\textbf{x}_i$ is a vector of concentrations of the molecular species in $R_i$ while $\Theta_i$ denotes the parameters associated with the reactions in $R_i$. We assume that the individual reactions in the pathway are governed by mass action or Michaelis-Menten kinetics \cite{klipp2005systems} (other types of reaction kinetics can be handled by our method as well). Hence each $f_i$ will be of the form: $f_i =\sum_{j=1}^{r_i}c_j g_j$, where $r_i$ denotes the number of reactions in $R_i$ and $c_j = +1$ ($-1$) if $x_i$ is a product (reactant) of the $j^{\mathrm{th}}$ reaction. In addition, $g_j$ are rational functions of the form
$g_j= \theta_{\alpha} x_i x_k$ (mass action) or $g_j=\theta_\kappa x_i / (\theta_{\kappa'} + x_i)$ (Michaelis-Menten) with $k \in \{1,2,\ldots,n\}$ and $\alpha, \kappa, \kappa' \in \{1,2,\ldots,m\}$. We are assuming here all the stoichiometric coefficients to be $1$   but non-unitary coefficients can be incorporated easily.

Fig.~\ref{example} shows a simple biochemical network depicting an enzyme catalyzed reaction and its mass action model as an ODE system. In this network, the enzyme $E$ binds reversibly to the substrate $S$ which then leads irreversibly to the generation of the product $P$ and release of the enzyme. The parameters $k_1$, $k_2$ and $k_3$ are rate constants that govern the rate of these reactions. The corresponding ODEs are shown in Fig.~\ref{example}(b).
\vspace{-12pt}
\begin{figure}[!ht]
\centering
\includegraphics[scale=0.45]{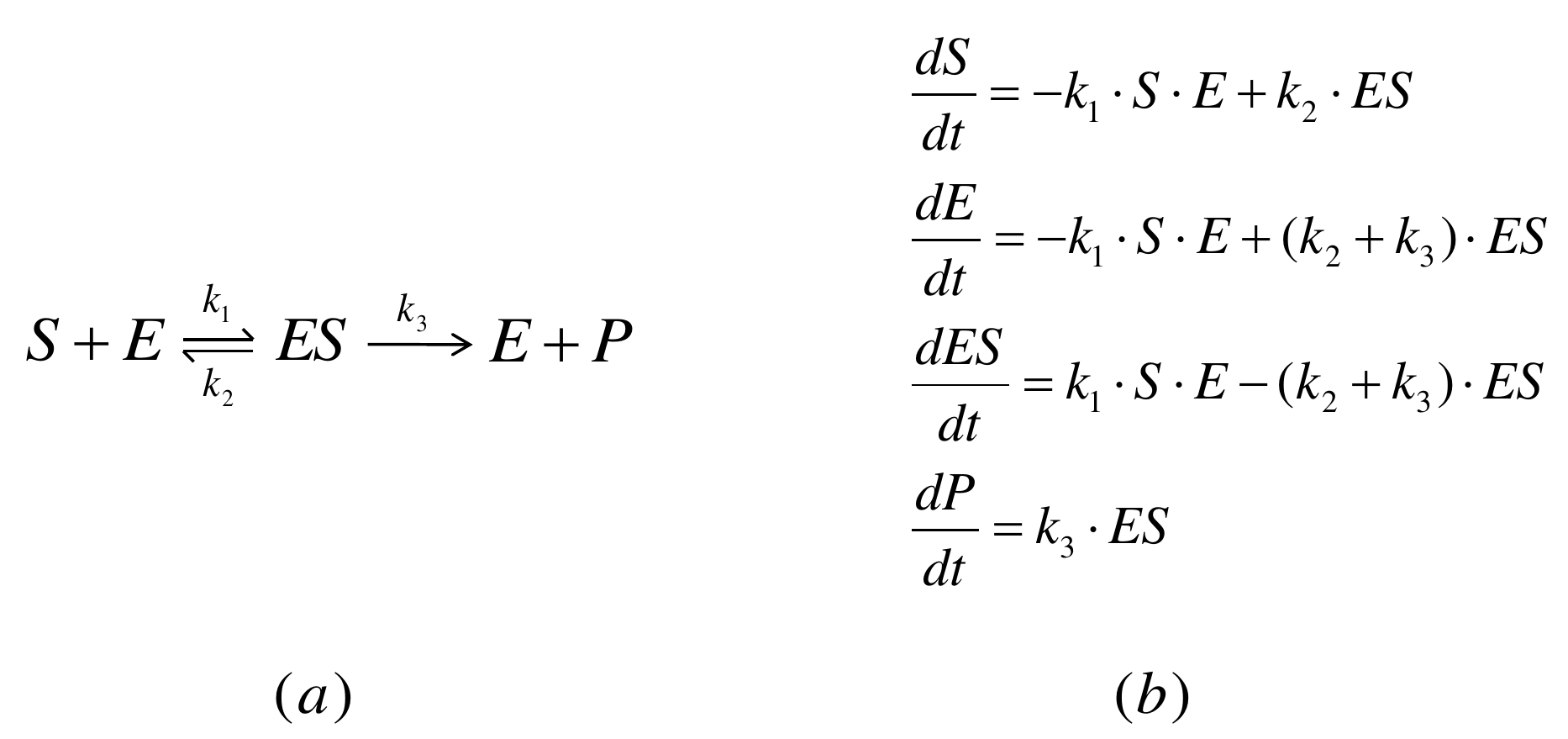}
\centering
\caption{A biochemical network and its ODE model}
\label{example}
\end{figure}

Each $x_i$ is to be viewed as a real-valued function of  time $t$ whose domain will be $\mathbb{R}_{+}$, the set of non-negative real numbers. We assume the range of values of $x_i(t)$ (corresponding to the concentration levels that can be attained by the species $x_i$) to be the interval $[L_i, U_i]$, where $L_i$ and $U_i$ are non-negative rational numbers with $L_i < U_i$.  We define $\vV = [L_1, U_1] \times \ldots \times [L_n, U_n]$. We let $\Theta = \{\theta_1, \theta_2, \ldots, \theta_m\} $ to be the set of all parameters and assume that $\theta _j$ takes values in $[L^j, U^j]$ with $1 \leq j \leq m$. We also define 
$\vW = [L^1, U^1] \times \ldots \times [L^m, U^m]$. 

To represent single-cell variability in the initial states, we associate an interval $[L^{init}_i, U^{init}_i]$ with each  $x_i$  such that $L_i \leq L^{init}_i < U^{init}_i \leq U_i$, and and interval $[L^j_{init}, U^j_{init}]$ with each parameter $\theta_j$ such that $L^j \le L^j_{init} < U^j_{init} \le U^j$. In what follows, an implicit and important assumption is that the value of a parameter does not evolve over time.
We define $INIT = (\prod_i [L^{init}_i, U^{init}_i]) \times (\prod_j [L^j_{init},U^j_{init}])$. We let $\vv \in \vV$ to range over $\prod_i [L^{init}_i, U^{init}_i]$ and $\ww \in \vW$ to range over $\prod_j [L^j_{init} , U^j_{init}]$.

\vspace{-0.5cm}
\subsection{Trajectories of the ODE system}
\vspace{-0.3cm}
The vector form of our ODE system is $d \textbf{x} /dt = F( \textbf{x} , \Theta )$.
Since we have assumed each reaction in $R_i$ is governed by  mass action or Michaelis-Menten~\cite{klipp2005systems} kinetics, it is justifiable to assume each $f_i$ is of class $C^1$ (i.e. continuously differentiable). As a result, $F: \vV \rightarrow \vV$ is also of class $C^1$. Consequently, for each $(\vv, \ww) \in INIT$ the ODE system will have a unique solution  $\textbf{X}_{\vv,\ww}(t)$ \cite{hirsch2012differential}. Furthermore, this solution will satisfy:  $\textbf{X}_{\vv,\ww}(0) = \vv$ and $\textbf{X}_{\vv,\ww}'(t) = F(\textbf{X}_{\vv,\ww}(t))$. It is also guaranteed that $\textbf{X}_{\vv,\ww}(t)$ is of class $C^0$ (i.e. continuous) \cite{hirsch2012differential} and therefore measurable. This will constitute the underpinning for our SMC procedure formulated in Sect. \ref{sec:smc-ode}.

It is helpful to define the flow $\Phi : \mathbb{R}_{+} \times \textbf{V} \times \textbf{W} \rightarrow \textbf{V}$ induced by the solution to the ODE system. $\Phi(t, \textbf{v}, \ww)$ is intuitively the state reached if the system starts at time $0$ with $\textbf{v}$ as the values of the variables and $\ww$ as the values of the parameters. The flow will be of class $C^0$ given by: $\Phi(t, \textbf{v}, \ww) = \textbf{X}_{\textbf{v},{\textbf{w}}}(t)$. This function  will satisfy 
  $\Phi(0,\textbf{v}, \ww) = \textbf{X}_{\textbf{v},\textbf{w}}(0) = \textbf{v}$
and $\partial(\Phi(t,\textbf{v}, \ww))/{\partial t} = F(\Phi(t,\textbf{v}, \ww))$ for all $t$.

We will also sometimes  work with the induced flow $\Phi_t : \vV \times \wW \ra \vV$ given by 
$\phi_t(\vv, \ww) = \Phi(t, \vv, \ww)$. Again $\Phi_t$ will be a $C^0$ function for each 
$t$ in $\mathbb{R}_{+}$.

The dynamics of interest of the ODE systems in our intended applications will be only up to a maximal time point $T$. Hence throughout what follows we fix such a positive rational $T$. 
Finally, a  \emph{trajectory} of the ODE system  is denoted $\sigma_{\vv,\ww}$ if it starts at time $0$ from $\vv \in \vV$, and uses $\ww$ as the vector of parameter values. It is the $C^0$ function $\sigma_{\vv,\ww}: [0, T] \rightarrow \vV$ that satisfies: $\sigma_{\vv,\ww}(t) =
X_{\vv, \ww}(t)$. We denote the the behavior of interest of the ODE system as $BEH$, which is the following set of trajectories:
\begin{equation*} 
BEH = \{ \sigma_{\vv,\ww} \mid (\vv,\ww) \in INIT\}.
\end{equation*}

\vspace{-0.5cm}
\section{Statistical model checking of ODE dynamics}
\label{sec:smc-ode}
\vspace{-0.3cm}
We now turn to the statistical analysis of the behavior of an ODE system. The key idea is to assume a probability distribution over $\INIT$ and use it to induce a distribution over $BEH$. 

\vspace{-0.5cm}
\subsection{BLTL}
\label{sec:smc_bltl}
\vspace{-0.3cm}
Our formulas will be interpreted at a finite set of time points $\mathcal{T} = \{0, 1, \ldots, T\}$. Discretization in time is justified since experimental observations are typically only available at a finite number of specific time points, and the fact that it allows expressing qualitative properties of interest. Going forward, we assume that an appropriate $\mathcal{T}$ has been chosen which includes the time points relevant to both the available quantitative experimental observations and the qualitative properties of interest.

Given these considerations, BLTL is a suitable a specification logic for our SMC procedure. The atomic propositions of the logic take the form $(i,\l,u)$ with $L_i \leq \l < u \leq U_i$, which is interpreted as ``the current level of $x_i$ is in the $[\l, u]$ interval''. We assume that a finite set of such atomic propositions  are fixed. 

The BLTL formulas are defined via: 
\begin{enumerate}[(i)]
	\item Every atomic proposition and the constants $\emph{true}$, $\emph{false}$ are BLTL formulas.
	\item If $\psi$, $\psi'$ are BLTL formulas then $\lnot \psi$ and $\psi \vee \psi'$  are also BLTL formulas.
	\item If $\psi$, $\psi'$ are BLTL formulas then $\psi \Ult  \psi'$  and $\psi \Ut  \psi'$ are BLTL formulas for any positive integer $t \leq T$.
\end{enumerate}

In addition to the derived propositional operators $\wedge$, $\supset$, $\equiv$, and the temporal operators $\Glt$,  $\Flt$, which are defined in the usual way, we have slightly strengthened BLTL to allow one to specify that a property will hold exactly $t$ time units from now. This enables the construction of properties from experimental time course data.

The semantics of the logic w.r.t. $\sigma, t \models \varphi$, where $\sigma$ is a trajectory in $BEH$ and $t \in \mathcal{T}$ are defined as follows.
\begin{itemize}
\item $\sigma, t \models (i, \l, u)$ iff $\l \leq \sigma(t)(i) \leq u$ where $\sigma(t)(i)$ is the $i^{\mathrm{th}}$ component of the $n$-dimensional vector $\sigma(t) \in \vV$.
\item $\lnot$ and  $\vee$ are interpreted in the usual way.
\item $\sigma, t \models \psi \Ulk \psi'$ iff there exists $k'$ such that $k' \leq k$,\,\, $t + k' \leq T$
and $\sigma, t + k' \models \psi'$. Further, $\sigma, t + k'' \models \psi$ for every $0 \leq k'' < k'$.
\item $\sigma, t \models \psi \Uk \psi'$ iff $t + k \leq T$ and $\sigma, t + k \models \psi'$. Further, $\sigma, t + k' \models \psi$ for every $0 \leq k' < k$.
\end{itemize}

We set $models(\psi) = \{\sigma \, \mid \, \sigma,0 \models \psi, \,\, \sigma \in BEH\}$.

Based on these semantics, we define statements of the form $P_{\ge r}(\psi)$, with the intended meaning ``the probability that a trajectory in $BEH$ is in $models(\psi)$ is at least $r$''. For such a statement to be meaningful  we need  a probability measure over $BEH$ that respects the dynamics of the ODEs system. The key observation here is that the initial (vector) value $\sigma(0)$ taken on at time $t = 0$ completely determines the trajectory $\sigma \in BEH$. We can therefore identify $BEH$ entirely with $\INIT$, the set of initial states. To make this relationship more explicit, we define the set $Models(\psi) \subseteq \INIT$ as
$(\vv, \ww ) \in Models(\psi)$ iff $\sigma_{\vv,\ww} \in models(\psi)$. 

Let $\mathcal{B}_{\INIT}$--from now on written as $\mathcal{B}$--be the $\sigma$-algebra generated by the $(n+m)$-dimensional 
open intervals contained in $\INIT$. We note that $\INIT$ is a member of $\mathcal{B}$. 

\begin{thm}
		(i) Let $\psi$ be a formula of BLTL and $t \in \mathcal{T}$. 
		Let $X_{\psi, t} \subseteq \INIT$ given by $\{(\vv,\ww) \mid \sigma_{\vv,\ww}, t \models \psi\}$. Then $X_{\psi, t} \in \mathcal{B}$.
		
		\noindent (ii) For every BLTL formula $\psi$, $Models(\psi)$ is a member of $\mathcal{B}$.
\end{thm}

\begin{proof}
	We establish the first part by structural induction $\psi$. Suppose $\psi = (i, l, u)$ is an atomic proposition.
	Let $I_{\psi} = \{\vv \mid l \leq \vv[i] \leq r \}$. Clearly $I_{\psi}$ is a measurable subset of $\vV$ in the usual sense. Recall from the previous section that 
	$\Phi_{s} : \textbf{V} \times \textbf{W} \rightarrow \textbf{V}$ is $C^0$ and hence 
	$\Phi_{s}^{-1}(I_{\psi}) $ will be a  subset of $\INIT$ and measurable in that it will be a member 
	of  $\mathcal{B}$  for each $s$ in $\mathbb{R}_{+}$.
	But then $\Phi_{t'}^{-1}(I_{\psi}) = X_{\psi, t}$ and this establishes the basis step. The arguments for negation and disjunction follow easily from the induction hypothesis since $\mathcal{B}$ is closed under complement and (countable) union. The $\Ult$ and 
	$\Ut$ cases follow from the induction hypothesis and the fact that $\mathcal{B}$ is closed under finite (in fact countable) intersection. \qed
	
\end{proof}

In what follows we let $P$ be the uniform probability distribution defined over $\mathcal{B}$ in the usual way assuming that there is no prior knowledge.  If such information is available it can be easily incorporated in the definition of $P$ of initial states.   
We shall say that $\mathcal{S}$, the system of ODEs, meets the specification $P_{\geq r} \psi$--and this is denoted $\mathcal{S} \models P_{\geq r} \psi$--iff $P(Models(\psi)) \geq r$, while $\mathcal{S} \models P_{\leq r'} \psi$ iff $P(Models(\psi)) \leq r'$. 

Below we refer to a formula of the form $P_{\geq r} \psi$ as a PBLTL formula.

\vspace{-0.5cm}
\subsection{Verifying PBLTL formulas using statistical model checking}
\vspace{-0.3cm}

In this section, we introduce a statistical method to decide whether a given model satisfies a property of the form $P_{\geq r} \psi$. The method produces an answer that comes with statistical guarantees, determined by a set of user-provided parameters. As an alternative to approximating the probability of satisfaction of $\psi$ directly~\cite{herault2004approximate}, we formulate a hypothesis test to decide whether $\mathcal{S} \models P_{\geq r} \psi$.

As described in~\cite{younes_statistical_2006}, such a test can be posed between a null hypothesis  ${\mathrm{H0} : p \ge r+\delta}$ and an alternative hypothesis ${\mathrm{H1} : p \leq r - \delta}$, where $p = P(Models(\psi))$. The parameter $\delta$ signifies an indifference region in which a decision between $\mathrm{H0}$ or $\mathrm{H1}$ cannot be made. Two additional parameters, $\alpha$ and $\beta$ establish the \emph{strength} of the statistical test by bounding the probability of a Type-I error (verifying a true property as false) and a Type-II error (verifying a false property as true), respectively. Even though the verification is approximate, these user-defined parameters allow setting statistical guarantees on confidence levels and error bounds.

The statistical test begins by taking a random sample from $\INIT$ as the initial state, starting from which the trajectory $\sigma_1$ is generated. The value $y_1$ representing the satisfaction of the property is then constructed as $y_1:=1$ if $\sigma, 0 \models \psi$, and $y_1:=0$ otherwise. This trajectory sampling step is repeated multiple times to generate a sequence of Bernoulli random variables $y_1, y_2, \ldots$. A key aspect of sequential testing is that based on the sequence of samples thus far, it allows deciding if the samples taken are sufficient to meet the required statistical guarantees or whether more samples need to be collected. The quantity $q_m$ is used as stopping criterion for the sequential sampling procedure and is calculated for each $m \geq 1$, after drawing $m$ samples as:
\begin{equation}
q_m = \frac{{[r-\delta]}^{(\sum_{i=1}^{m}y_i)}{[1-[r-\delta]]}^{(m - \sum_{i=1}^{m}y_i)}}{{[r+\delta]}^{(\sum_{i=1}^{m}y_i)}{[1-[r+\delta]]}^{(m - \sum_{i=1}^{m}y_i)}} \quad . \label{eq:fm}
\end{equation}
We stop sampling and accept hypothesis $\mathrm{H0}$ if $q_m \le \widehat{B}$, and accept $\mathrm{H1}$ if $q_m \ge \widehat{A}$. Otherwise, if neither sopping condition is satisfied, we continue sampling. The constants $\widehat{A}$ and $\widehat{B}$ are chosen such that a test of strength $(\alpha,\beta)$ is achieved. An approximation that has been shown to meet the statistical guarantees under most choices of parameters in practice is described in \cite{younes_statistical_2006} as $\widehat{A}=\frac{1-\beta }{\alpha }$ and $\widehat{B}=\frac{\beta }{1- \alpha }$.

Strictly speaking, statistical model checking requires a stochastic dynamical system to define the trajectories. 
In the present setting it will be easy to define a DTMC by transferring the probability measure $P$ to a $\sigma$-algebra over $BEH$ and use it to define a $DTMC$. It is however unnecessary since it will yield the same results as the procedure described above.

\vspace{-0.5cm}
\subsection{SMC based parameter estimation for a single system of ODEs}
\vspace{-0.3cm}
We now describe a method for parameter estimation that builds on the SMC procedure.
We aim to represent two types of properties using BLTL formulas: experimental data and qualitative trends. We assume that experimental data allotted to be used for parameter estimation is available for the subset of variables $O \subseteq \{x_1, x_2, \ldots, x_k\}$. We also assume that the measured values of species $x_i$ at time $t$ are reported as an interval $[\ell^i_t, u^i_t]$ for each $t \in \mathcal{T}_i$ where $\mathcal{T}_i = \{\tau^i_1, \tau^i_2, \ldots, \tau^i_{T_i}\}$ is the set of time points where measurement data is available for $x_i$. In each case, the interval $[\ell^i_t, u^i_t]$ reflects the cell-population based nature, and noisiness of the experimental data. For each $t \in \mathcal{T}_i$, the property $\psi^t_i = \Ft (i, \ell^i_t, u^i_t)$ corresponds to a single measurement from which we construct $\psi^{i}_{exp} = \bigwedge_{t \in \mathcal{T}_i} \psi^t_i$ as a property over all time points. We then combine these properties over all species as $\psi_{exp} = \bigwedge_{i \in O} \psi^i_{exp}$. The case where there are multiple experimental conditions under which $x_i$ has been measured can be handled with an obvious extension of the above encoding scheme.

In addition to data, qualitative dynamical trends for some of the species in the pathway will also be available, often from the literature. For example, we may know that one of the species in the model shows oscillatory behavior with given properties. Similarly, a species may be know to show sustained activation in which it increases to a high level early on and stays at that level without returning to its original value. This type of qualitative information can also be formulated as a set of BLTL formulas which we will refer to as \emph{trend} formulas. 

Given both types of properties, we construct a PBLTL formula $P_{\geq r}(\psi_{exp} \wedge \psi_{qlty})$, where $r$ is the confidence level with which we wish to evaluate the goodness of fit of a given set of parameters to qualitative trends and experimental data. We also choose values for the $(\alpha, \beta)$, and $\delta$ parameters, determining the strength of the test and the indifference region, respectively.

Since both $\psi_{exp}$ and $\psi_{qlty}$ are conjunctions, we will exploit the fact that their terms can be tested separately. To account for this separation, following~\cite{younes_statistical_2006}, we choose the strength of each of the individual tests to be $(\frac{\alpha}{J},\beta)$, where $J$ is the total number of conjuncts in the specification. This guarantees that the strength of the test remains $(\alpha, \beta)$ overall. Testing terms separately also allows us to use the results of individual statistical tests for computing the objective function associated with the global search strategy introduced below.

We assume that the model parameters can be separated into two sets: parameters whose values are unknown, represented as $\Theta_u = \{\theta_1, \theta_2, \ldots, \theta_K\}$, which need to be estimated, and parameters whose nominal values are known and are not subject to variability across a population of cells. Further, we assume that the nominal initial concentrations of species is given as an interval representing the range of their fluctuations as $[L^{init}_i, U^{init}_i]$ for each variable $x_i$. We also choose a constant $\delta''$  to define an interval around the current estimate of the unknown parameter values such that for a parameter vector $\ww \in \prod_{1 \leq j \leq K} [L^j , U^j]$, its value is assumed to vary in the range $[\ww(j) - \delta'', \ww(j) + \delta'']$. Then by setting $L^j_{init, \ww} = \ww(j) - \delta''$ and $U^j_{init, \ww} = \ww(j) + \delta''$, we can define $\INIT_{\ww} = (\prod_i [L^{init}_i, U^{init}_i]) \times (\prod_j [L^j_{init, \ww}  U^j_{init, \ww}])$. The set of trajectories $BEH_{\ww}$ is defined accordingly.

When assessing the goodness of any given $\ww$, we run the SMC procedure--using $\INIT_{\ww}$ instead of $\INIT$--to evaluate $P_{\geq r}(\psi_{exp} \wedge \psi_{qlty})$. We then apply an objective function to the various components of this test and use the objective function value to guide an iterative search strategy over the possible values of $\ww$. But the satisfaction of properties derived from experimental data, and those representing qualitative properties are part of the objective function. In the term quantifying fit to data, statistical tests are evaluated species-wise and then composed by summing over the normalized contribution of each species. The term representing fit to qualitative properties evaluates the number of statistical tests with $\ww$ that resulted in acceptance of the null hypothesis (i.e. the desired outcome).

Let $J^i_{exp}$ ($= T_i$) be the number of conjuncts in $\psi^i_{exp}$, and $J_{qlty}$ the number of conjuncts in $\psi_{qlty}$. Let $J^{i,+}_{exp}(\ww)$ be the number of formulas of the form $\psi^t_i$ (a conjunct in $\psi^i_{exp}$) such that the statistical test for $P_{\geq r}(\psi^t_i)$ accepts the null hypothesis (that is, $P_{\geq r}(\psi^t_i)$ holds) with strength $(\frac{\alpha}{J},\beta)$, where $J = \sum_{i \in O} J^i_{exp} + J_{qlty}$. Similarly, let  $J^{+}_{qlty}(\ww)$ be the number of conjuncts of the form $\psi_{\ell, qlty}$ in n $\psi_{qlty}$ that pass the statistical test $P_{\geq r}(\psi_{\ell, qlty})$ with strength $(\frac{\alpha}{J},\beta)$. Then the objective function $\mathbf{G}(\ww)$ is calculated as:

\begin{equation}
\mathbf{G}(\ww) = J^{+}_{qlty}({\ww}) + \sum_{i \in O} \frac{J^{i,+}_{exp}}{J^i_{exp}}.
\label{objval}
\end{equation}
To summarize, the goodness to fit of $\ww$ is quantified via the number of qualitative properties that are satisfied with it, and the normalized number of data points with which there is agreement within chosen bounds. The form of the objective function also implies that we do not necessarily require that the dynamics predicted by $\ww$ must fit every data point and qualitative property--a property which helps avoid over-fitting the model.

Given the objective function $\mathbf{G}(\ww)$, we use a search algorithm to maximize its value over the possible space of parameters. Our method is orthogonal to the specific choice of search strategy, and can be used in conjunction with both global and local search algorithms. Global search algorithms including Stochastic Ranking Evolutionary Strategy (SRES) \citep{sres} and Genetic Algorithms (GA) \citep{goldberg1989genetic} are more demanding computationally than local approaches but have been shown to be better at avoiding local minima. With global search methods, a population of candidate parameter vectors is maintained in each round of search (rounds are often called \emph{generations}). In our work, we apply the SRES algorithm due to the fact that it has been applied successfully on large pathway models \citep{moles_parameter_2003}.
\vspace{-0.6cm}
\section{Extension of SMC to hybrid automata}
\vspace{-0.3cm}
\label{sec:hybrid-model}
As mentioned in the introduction we extended the SMC based analysis to the much richer setting of hybrid automata. What follows is a condensed account of this extension; derived from the more detailed and complete presentation in \cite{gyori2015}. 
A hybrid automaton has multiple modes of operation and in each mode the system variables evolve continuously according to an associated system of ODEs. When the system enters a designated portion of the state space described by guards associated with the mode transitions, the system will move to a new mode instantaneously where it will start to evolve according to the ODEs associated with the new mode. 
 
Let $\{x_i\}_{i=1}^n$ be $n$-real valued variables viewed as functions of time $x_i(t)$ with  $t \in \mathbb{R}_{+}$. A valuation of $\{x_i\}_{i=1}^n$  is  $\bv \in \mathbb{R}^n$ with $\bv(i) \in \mathbb{R}$ representing the value of $x_i$. The language of \emph{guards} is given by: 
\begin{enumerate}
\item[(i)] $a < x_i $ and $x_i < b$ are guards where $a, b$ are rational numbers. 
\item[(ii)] If $g$ and $g'$ are guards then $g \wedge g'$ and $g \lor g'$ are also guards.
\end{enumerate}

Let $\mathcal{G}$  denote the set of guards. We define the notion of the valuation $\bv$ satisfying  the guard $g$)--denoted $\bv \models g$--as follows: $\bv \models  a < x_i $ iff $a < \bv(i)$; similarly for $x_i < b$. The conjunction and disjunction cases are as expected.   We also define  $\parallel\!\! g\!\! \parallel = \{ \bv \mid \bv \models g\}$. For every guard $g$ we note that $\parallel\!\! g \!\!\parallel$ is an open subset of $\Re^n$. We often abbreviate  $\parallel\!\! g\!\! \parallel$ as $g$.

\begin{definition}\label{def:SHA}
A hybrid automaton is a structure
$H=(Q, q_{in},$ $ \{F_q(\bx)\}_{q\in Q},$ $ \mathcal{G},$ $ \ra,$ $ \INIT)$,
where
\begin{itemize}
\item $Q$ is a finite set of \emph{modes} with $q_{in} \in Q$ as the \emph{initial} mode.
\item $d\vx/dt = F_q(\vx)$ is a system of ODEs for each $q\in Q$. Here $\vx = (x_1, x_2, \ldots,x_n)$ and $F_q = (f^1_q(\vx), f^2_q(\vx),$ $\ldots, f^n_q(\vx))$. We require $f_q$ to be Lipschitz continuous for each $q$.
\item $\ra \subseteq (Q, \mathcal{G}, Q)$ is the mode transition relation. In what follows $(q,g,q')\in \ra$ will be often written it as $q\move{g}q'$.

\item $\INIT = (L_1, U_1) \times  (L_2, U_2) \ldots \times (L_n, U_n)$ is the set of initial states where $L_i < U_i$ and $L_i, U_i$ are rationals.
\end{itemize}
\end{definition}
In Fig.~\ref{heart} we show a hybrid automaton model of a cardiac cell. This automaton has $4$ modes and (by coincidence) $4$  system variables. 
In order to limit the notational overhead we are assuming here that the behavioral variability arises  solely through intervals of values for the initial concentrations of the 
molecular species. Further, we are not associating other features such as invariant conditions or reset conditions \cite{henzinger-lics-survey}. They can be easily handled with some additional effort. We also note we do not require the guards associated with the outgoing transitions of a mode to be disjoint.
In other words we allow non-determinism which will be converted to stochasticity via our definition of the transition probabilities of the Markov chain approximation explained below. This is also brought out in the algorithm for sampling trajectories presented in Sect. \ref{sec:smc_hybrid}.

We choose a time discretization as $t = 0, \Delta, 2 \Delta, \ldots$ where $\Delta$ is a suitably chosen unit time interval. The system's states are assumed to be observed at these discrete time points. We also assume that the number of mode changes that can occur between two consecutive discrete time points is bounded. In biological pathways this is a realistic assumption. In fact, we will assume, for convenience, that in each unit time interval, \emph{no more than one} mode change takes place. Clearly, $\Delta$ can be chosen in multiple ways to meet this requirement, hence it must be chosen carefully. Our method can be easily extended to work with a bounded number of mode transitions occurring in a $\Delta$  interval.  This would, however, complicate the notation and make the main idea less clear--hence our stronger assumption. We further assume that $\Delta = 1$ for technical convenience. Hence, we will use $\{0, 1, 2, \ldots\}$ as the set of discrete time points.

Since $F_q(\vx)$ is Lipschitz continuous the ODE system $d\vx/dt = F_q(\vx)$ has a unique solution  $Z_{q, \vv}(t)$
for each mode $q$ and for each initial value $\vv \in \mathbb{R}^n$ \cite{hirsch2012differential}. This guarantees that $Z_{q, \vv}(t)$ is also Lipschitz continuous and therefore measurable \cite{hirsch2012differential}. 

The \emph{flow} $\Phi_q : (0, 1) \times \mathbb{R}^n \ra \mathbb{R}^n$ in a unit interval is given by $\Phi_q(t, \vv) = Z_{q, \vv}(t)$. Again, $\Phi_q$ will also be Lipschitz continuous.

A (finite) \emph{trajectory} is a sequence
$\tau = (q_0, \vv_0) \, (q_1, \vv_1) \, \ldots (q_k, \vv_k)$ such that for $0 \leq j < k$  the following conditions are satisfied: 

\begin{enumerate}[(i)]
	\item For $0 \leq j < k$, $q_{j} \stackrel{g_j}{\ra} q_{j +1}$ for some guard $g_j$
	\item There exists $t \in (0, 1)$ such that:
	      \begin{itemize}
	      	\item $\Phi_{q_{j}, t}(\vv_{j}) \in g$
	      	\item $\vv_{j + 1} = \Phi_{q_{j +1}, 1-t}(\Phi_{q_{j}, t}(\vv_{j}))$
	      \end{itemize}
\end{enumerate}
  
The trajectory  $\tau$ defined above is said to \emph{start} from $q_0$ and \emph{end} in $q_k$ with $\vv_0$ as its initial state, and $\vv_k$ as its final value state. We denote the set of all finite trajectories that start from an initial value state in $\INIT$ in the initial mode $q_{in}$ as $TRJ$.

\vspace{-0.5cm}
\subsection{The Markov chain associated with a hybrid automaton}\label{sec:mc}
\vspace{-0.3cm}
As before our goal is to assume a probability distribution over $\INIT$ and use it to derive a stochastic version of  
the dynamics of the hybrid automaton. This will open the door for performing statistical analysis. As in the case of a single ODE system we begin assigning a probability distribution over $\INIT$. For convenience we will assume this to be the uniform distribution $\textbf{P}_{\INIT}$ in what follows. In the present much richer setting it will be difficult to directly lift $\textbf{P}_{\INIT}$ to a distribution over the trajectories of $H$. Instead we will first associate a (discrete time) Markov chain $M_H$ with $H$ and  establish a strong relationship between the behaviors of $H$ and $M_H$ relative to (all) properties specified as BLTL formulas. We then develop an SMC procedure for $M_H$ which will amount to an approximate analysis of $H$ but with statistical guarantees.

$M_H$ will be an acyclic finitely branching infinite state DTMC. We shall describe informally how it is built up inductively.
The technical details can be found in \cite{supplementary}.  We start with $(q_{in}, \INIT, \textbf{P}_{\INIT})$ as the initial state of $M_H$. We say that $q_{in}$ is the current mode and $q_{in}$ (viewed as a string  of length $1$) is the path taken from $q_{in}$ to reach this mode. Furthermore $\INIT$ is the set of all the possible initial states that can be taken with 
$\textbf{P}_{\INIT}$ as the distribution over this set. 

Assume inductively that $(\rho, X, \textbf{P}_X)$ is a state of $M_H$  with $q$ as the current mode, 
$\rho$ is the path taken 
from the initial state of $M_H$ to reach this state, $X$ is the set of all the value states that are 
possible in this state and $\textbf{P}_X$ the  probability distribution over $X$. Suppose there are $m$ outgoing transitions $q\move{g_1}q_1,\ldots, q\move{g_m}q_m$ from $q$ in $H$ (Fig. \ref{fig:mc} shows this inductive step).

\begin{figure}[htb]
	\centering
	\begin{tikzpicture}[node distance=0.5cm and 0.5cm,shorten >=1pt,auto,>=stealth',
	mynode/.style={ellipse,draw,font=\normalsize,text width=2.3cm,minimum height=0.9cm,minimum width=1.5cm,align=center}, every node/.style={scale=0.7}]
	\node[mynode](init) at (0,0) {$(q_{in},\INIT,\textbf{P}_\INIT)$};
	\node[draw=none](u1) at (-4,-1.1) {};
	\node[draw=none](u2) at (-2,-1.1) {};
	\node[draw=none](d1) at (2,-1.1) {};
	\node[draw=none](d2) at (4,-1.1) {};
	\node[mynode](r) at (0,-1.1) {$(\rho,X,\textbf{P}_X)$};
	\node[mynode](ru1) at (-4,-2.2) {$(\rho q_1,X_1,\textbf{P}_{X_1})$};
	\node[draw=none](ru2) at (-2,-2.2) {$\ldots$};
	\node[mynode](rr) at (0,-2.2) {$(\rho q_j,X_j,\textbf{P}_{X_j})$};
	\node[draw=none](rd2) at (2,-2.2) {$\ldots$};
	\node[mynode](rd1) at (4,-2.2) {$(\rho q_m,X_m,\textbf{P}_{X_m})$};
	
	\draw[->,decorate,decoration={snake,amplitude=.4mm,segment length=2mm,post length=1mm}](init) to node[right]{$\rho$} (r);
	\draw[->,dashed] (init) to (u1);
	\draw[->,dashed] (init) to (u2);
	\draw[->,dashed] (init) to (d1);
	\draw[->,dashed] (init) to (d2);
	\draw[->] (r) to (ru1);
	\draw[->,dashed] (r) to (ru2);
	\draw[->,dashed] (r) to (rr);
	\draw[->] (r) to node[right,rotate=60, xshift=-13pt]{\Huge{$\times$}} (rd1);
	\draw[->,dashed] (r) to (rd2);
	\end{tikzpicture}
	\caption{The Markov chain construction. The edge between $(\rho,X,\textbf{P}_X)$ and $(\rho q_m,X_m,\textbf{P}_{X_m})$ is marked with a `{\large $\times$}', representing the case where the probability of a transition is 0 since $X_m$ has measure 0. Thus, $(\rho q_m,X_m,\textbf{P}_{X_m})$  will not be a state of the Markov chain. Reprinted by permission from Springer \cite{gyori2015}.}
	\label{fig:mc}
\end{figure}
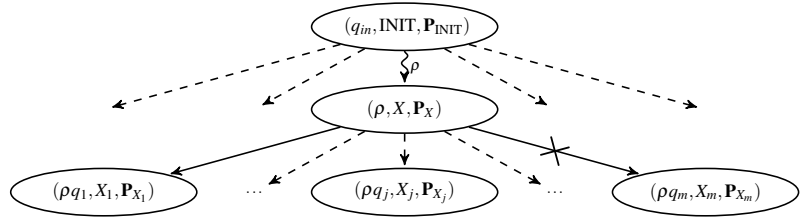

Then for $1 \leq j \leq m$ we define the triples $(\rho q_j, X_j, \textbf{P}_{X_j})$ as follows. For convenience, let $j$  range over $\{1, 2, \ldots, m\}$ in what follows.

We define $\tm_j(\vv) \subseteq (0, 1)$ to be 
time points $t$  in $(0, 1)$ at which the guard $g_j$ is satisfied if the system starts in mode $q$  from value $\vv$ and evolves according to the mode dynamics of $q$ up to time $t$. 

For each $j$, we define $X_j$ to be the set of state values obtained by starting at time $k$ from some value $\vv \in X$ , then following the dynamics of $q$ to evolve up to $k + t$, switching instantaneously to mode $q_j$ at that time point, and then following the dynamics of the new mode $q_j$ up to time $k+1$.

Finally  we define $\bP_{X_j}$, the probability distribution over $X_j$ as follows. Informally, the probability mass of $\by$ in $X_j$ under $\bP_{X_j}$ is the sum of the  probability masses  
of points $\textbf{z}$ in $X$ under $\bP_X$ from which one can evolve according to the dynamics of $q$ up to some time point
$t$  in $\tm_j(\textbf{z})$, transition to mode $q_j$ and evolve for time $1-t$ to reach $\by$. We next define the transition probability $M_H((\rho, X, \bP_X), (\rho q_j, X_j, \bP_{X_j}))$ to be proportional to the sum of $\tm_j(\textbf{v})$ taken over $\textbf{v} \in \textbf{X}$.  The precise details are more involved; for instance one must use integrals instead of  summations \cite{supplementary}.

\vspace{-0.5cm}
\subsection{Relating the behaviors of \large{\bf $H$} and \large{\bf $M_H$} using BLTL}\label{pbltl}
\vspace{-0.3cm}
As before we shall use BLTL as a specification language for properties bounded in time, and use it as a basis for relating the behaviors of $H$ and $M_H$. For convenience, in what follows, we will write $M$ instead of $M_H$.

We assume a finite set of atomic propositions $AP$ and a valuation function $Kr: Q \ra 2^{AP}$ with BLTL formulas defined as in Section \ref{sec:smc_bltl}. In the setting of hybrid automata we will be reasoning mainly about properties regarding the sequences of modes generated by the automaton. This is so since the results in \cite{ASTY}--based on constructions developed in \cite{HKPV}--show that in the presence of the non-linear  dynamics we are dealing with here, even simple quantitative propositions will result in undecidability of mode reachability. Furthermore as the properties specified in our various case studies show (see also \cite{supplementary}), many biologically relevant properties can be captured purely in terms of mode sequences. 

As before, we assume that there exists a maximum time point $K$ up to which the behavior of the system is assumed to be of interest. For a BLTL formula $\psi$, we recall there is a constant $K_{\psi}$ depending only on $\psi$ such that evaluating an execution trace up to at most $K_{\psi}$ suffices to determine the satisfaction of $\psi$~\cite{TACAS99}. Accordingly, we assume that a sufficiently high value of $K$ has been chosen such that it can handle all the specifications of interest. We denote by $TRJ^{K+1}$ the set of trajectories of length $K + 1$. It represents the behavior of $H$ of interest.

Turning to the corresponding notion for $M$, a finite path in $M$ is  a sequence $\eta_0 \eta_1 \ldots \eta_k$ such that $\eta_j \in \Upsilon$ for $0 \leq j \leq k$. In addition,  for $0 \leq j < k$, $\eta_j \stackrel{p_j}{\Rightarrow} \eta_{j + 1}$ for some $p_j \in (0, 1]$. We say such a path starts from $\eta_0$. Its length is $k + 1$. Then, $\paths_M$ is defined as the set of finite paths starting from the initial state of $M$ and $\paths^{K + 1}_M$ to be the set of paths in $\paths_M$ of length $K + 1$.

\runinhead{The trajectory semantics}
Let $\psi$ a BLTL formula and $\tau = (q_0, \vv_0)$ $(q_1, \vv_1)$ $\ldots$  $(q_k, \vv_k)$ a finite trajectory and $0\leq j \leq K$. Then $\tau, j \models_H \psi$  in case:
\begin{itemize}
	\item If $A$ is an atomic proposition then $\tau, j \models_H  A$ iff $A \in Kr(q_j)$.
	\item The cases $\lnot$ and  $\vee$ are handled in the expected way.
	\item $\tau,j \models_H \psi \U^{\leq {\ell}} \psi'$ iff there exists $m$ such that $m \leq \ell$ and $j + m \leq k$ and $\tau, (j + m) \models_H\psi'$. Furthermore, for every $0 \leq m' < m$ it is the case $\tau,(j + m') \models_H \psi$ .
\end{itemize}

The notion $models_H(\psi) \subseteq TRJ^{K + 1}$ is then defined via: $\tau \in models_H(\psi)$ iff $\tau, 0 \models_H \psi$. We assert $H$ \emph{meets the specification} $\psi$ -denoted $H \models \psi$- iff $models_H(\psi) = TRJ^{K + 1}$.

\runinhead{The Markov chain semantics}
Suppose $\pi= \eta_0 \eta_1 \ldots \eta_k$ is a path in $M$ where for $0 \leq j \leq k$ we have  $\eta_j = (\rho q_j, X_j, \textbf{P}_{X_j})$.
Let $\psi$ be a BLTL formula and $0 \leq j \leq k$.  Then  $\pi, j \models_M \psi$ iff:
\begin{itemize}
	\item If $A$ is an atomic proposition then $\pi, j \models_M  A$ iff $A \in Kr(q_j)$.
	\item The remaining clauses are just as they were for  $\models_H$. 
\end{itemize}

We now define $models_M(\psi) \subseteq \paths^{K + 1}_M$ as follows: $\pi \in models_M(\psi)$ iff $\pi, 0 \models_M \psi$.
The probability of a formula being satisfied in $M$ can now be defined. First  if $\pi = \eta_0 \eta_1 \ldots \eta_K$ is in $\paths^{K + 1}_M$ then $\Prob(\pi) = \prod_{0 \leq \ell < K} p_{\ell}$,
where $\eta_{\ell} \stackrel{p_{\ell}}{\Rightarrow} \eta_{\ell + 1}$ for $0 \leq \ell < K$. Now we define:
\begin{equation*}
	\Prob(models_M(\psi)) = \sum_{\pi \in models_M(\psi)} \Prob(\pi).
\end{equation*}
$\mc \models \psi$ will  denote $\Prob(models_M(\psi)) = 1$. As usual we will write $\Prob_{\geq p}(\psi)$ instead of $\Prob(models_M(\psi))$ $ \geq p$. It is useful to note that $Pr(\pi) > 0$ for every $\pi \in models_M(\psi)$. Moreover $\sum_{\pi \in models_M(\psi)} \Prob(\pi) \leq 1$. This in turn leads to $\Prob_{\geq 1}(\psi)$ iff $ models_M(\psi) = \paths^{K + 1}_M$ iff $M \models \psi$.

We note that $LTL$ was interpreted over Markov chains in \cite{vardi85}. However, this was a qualitative semantics in that only the notion of an $LTL$ formula being satisfied with probability $1$ by a Markov chain was defined. Here we deal with all probabilities.

\runinhead{The correspondence result}
Our aim is to establish  that $H$ meets the specification $\psi$ iff $\Prob_{\geq 1}(\psi)$. To do so,  let $\pi= \eta_0 \eta_1 \ldots \eta_k$ be a path in $M$ with $\eta_j = (q_0q_1\ldots q_j, X_j, \textbf{P}_{X_j})$ for $0 \leq j \leq k$. Let $\tau = (q'_0, \vv_0)$ $(q'_1, \vv_1) \ldots(q'_{k'}, \vv_{k'})$ be a trajectory. Then define $\pi$ and $\tau$ to be  \emph{compatible} just in case $k = k'$ and $q_j = q'_j$ and $\vv_j \in X_j$ for $0 \leq j \leq k$. The following lemma 
whose proof can be found in \cite{supplementary} easily leads to the correspondence result.

\begin{lemma}
	\label{lemmapbltl}
	\begin{enumerate}
		\item Suppose the path $\pi= \eta_0 \eta_1 \ldots \eta_k$ in $M$ and the trajectory $\tau = (q_0, \vv_0)$ $(q_1, \vv_1) \ldots(q_k, \vv_k)$ are compatible. Let $0 \leq j \leq k$ and $\psi$ be a BLTL formula. Then $\pi, j \models_M \psi$ iff $\tau, j \models_H \psi$.
		\item Suppose $\pi$ is a path in M. Then there exists a trajectory $\tau$  such that $\pi$ and $\tau$ are compatible. Furthermore if $\pi \in \paths_M$ then $\tau \in TRJ$.
		\item Suppose $\tau$ is a trajectory. Then there exists a path $\pi$ in $M$ such that $\tau$ and $\pi$ are compatible. Furthermore if $\tau \in TRJ$ then $\pi \in \paths_M$.
	\end{enumerate}
\end{lemma}

We then have:

\begin{theorem}
	\label{thm:main}
	$H \models \psi$ iff $M \models \psi$.
\end{theorem}
\begin{proof}
	Assume that  $H$ does not meet the specification $\psi$. This implies the existence of $\tau \in TRJ^{K+1}$ such
	that $\tau, 0 \not\models_H \psi$. By point (3) of Lemma~\ref{lemmapbltl} there must exist $\pi \in \paths^{K+ 1}_M$
	which is compatible with $\tau$. By point (1) of Lemma~\ref{lemmapbltl} we then have $\pi \notin models_M(\psi)$ which at once implies $Pr_{< 1}(\psi)$.
	
	Next assume $Pr_{< 1}(\psi)$. This implies there exists $\pi \in \paths^{K + 1}$ such that $\pi, 0 \not\models_M \psi$.
	By point (2) of Lemma~\ref{lemmapbltl} there must exist $\tau \in TRJ^{K + 1}$ that is compatible with $\pi$. This implies $\tau, 0 \not\models_H \psi$ by point (1) of Lemma~\ref{lemmapbltl}. This in turn implies that $H$ does not meet the specification $\psi$.
\end{proof}
\runinhead{The SMC procedure}\label{sec:smc_hybrid}
To check whether $H$ meets the specification $\psi$, we will equivalently check  whether $\Prob_{\ge 1}(\psi)$ on $M$.
As mentioned in Sect. \ref{sec:introduction} however even a finite initial portion of  $M$ cannot be constructed explicitly; to do so one will
need the solutions to the ODEs which will not be available.
We overcome this  by using randomly generated trajectories of $H$ in such a way that they can serve as proxies for randomly generated paths of $M$ and in this light submit them to a sequential hypothesis test. This will enable us to decide --with bounds on the rate of error--whether the hypothesis $H_0$ corresponding to $\Prob_{\ge 1}(\psi)$ holds or whether $H_1$ corresponding to $\Prob_{< 1-\delta}(\psi)$ holds. Our trajectory sampling procedure is described as Algorithm \ref{algtrajsim}.
\vspace{-0.2in}
\begin{algorithm}[H]
	\scriptsize
	\renewcommand{\thealgorithm}{1}
	\caption{Trajectory simulation}
	\label{algtrajsim}
	Input: Hybrid automaton $H=(Q, q_{in}, \{F_q(\bx)\}_{q\in Q}, \mathcal{G}, \ra, \INIT)$, maximum time step $K$.
	
	Output: Trajectory $\tau$
	\begin{algorithmic}[1]
		\State Sample $\vv_0$ from $\INIT$ uniformly, set $q_0:=q_{in}$ and $\tau:=(q_0,\vv_0)$.
		\For{$k:=1 \ldots K$}
		\State Generate time points  $T:=\{t_1,\ldots,t_J\}$ uniformly in $(0,1)$.
		\State Simulate $\vv^j := \Phi_{q_{k-1}}(t_j,\vv_{k-1})$, for $j\in\{1,\ldots,J\}$
		\State Let $\widehat{\tm}_j:=\{t \in T: \vv^j \in g_j\}$ be the time points where $g_j$ is enabled.
		\State Pick $g_{\ell}$ randomly according to probabilities $\{p_j := |\widehat{\tm}_j| / \sum_{i=1}^m |\widehat{\tm}_i|\}$.
		\State Pick $t_{\ell}$ uniformly at random from $\widehat{\tm}_{\ell}$.
		\State Simulate $\vv':=\Phi_{q'}(1-t_{\ell},\vv^{\ell})$, where $q'$ is the target of $g_{\ell}$.
		\State Set $q_k := q'$, $\vv_k := \vv'$, and extend $\tau := (q_0,\vv_0)\ldots(q_k,\vv_k)$.
		\EndFor
		\State \textbf{return} $\tau$
	\end{algorithmic}
\end{algorithm}
\vspace{-0.2in}

The proof of  correctness of this sampling procedure can be found in \cite{supplementary}.

We can now construct Algorithm \ref{alghyptest} which carries out the sequential hypothesis test. It uses Algorithm \ref{algtrajsim} to repeatedly generate a random trajectory, and decide between  $H_0$ and $H_1$ after a finite number of samples. The key parameters for the algorithm are the user-defined false positive rate $\alpha$, and the indifference interval $\delta$.
\vspace{-0.2in}
\begin{algorithm}[H]
	\scriptsize
	\renewcommand{\thealgorithm}{2}
	\caption{Sequential hypothesis test}
	\label{alghyptest}
	Input: BLTL property $\psi$, indifference parameter $\delta$, false positive bound $\alpha$.
	
	Output: $H_0$ or $H_1$.
	\begin{algorithmic}[1]
		\State Set $N := \lceil \log \alpha/\log (1-\delta) \rceil$
		\For{$i:=1 \ldots N$}
		\State Generate a random trajectory $\tau$ using Algorithm \ref{algtrajsim}
		\State \textbf{if} {$\tau, 0 \models^H \psi$} \textbf{ then } Continue
		\State \textbf{else return} $H_1$
		\EndFor
		\State \textbf{return} $H_0$
	\end{algorithmic}
\end{algorithm}
\vspace{-0.3in}

The fact that Algorithm \ref{alghyptest} has the required accuracy is shown in the following theorem.
\begin{theorem}
	\label{thm:smc} If $H_0$ is true then the probability of choosing $H_1$ (false negative) is $0$. On the other hand  suppose $N \ge \log\alpha/\log (1-\delta)$. Then the probability of choosing $H_0$ when $H_1$ is true (false positive) is less or equal to $\alpha$.
\end{theorem}

\begin{proof}
	The first part is trivial. To see the second part holds, assume $H_1$ is true, which implies that $\Prob_{<1-\delta}(\psi)$. We also know that the probability of $N$ sampled trajectories all satisfying $\psi$, thus leading to $H_0$, a false positive, is at most $(1-\delta)^N$. This implies  $\alpha \le (1-\delta)^N$, which at once leads to  $N \ge \log\alpha/\log (1-\delta)$. \qed
\end{proof}

Due to this result we use choose the sample size using the formula $N := \lceil \log\alpha / \log (1-\delta) \rceil$. As an example, if $\alpha = 0.01$ and $\delta = 0.05$, we get $N = 90$, while for $\alpha = 0.001$ and $\delta = 0.01$, we get $N = 688$.

\vspace{-0.5cm}
\section{Applications} 
\vspace{-0.3cm}
\label{sec:results}
The performance of our methods has been extensively tested in \cite{palaniappan2013,gyori2015,Ramanathan2015} using models taken from the literature \cite{elowitz2000synthetic,pourquie08,thrombin,brown2004statistical} (see Supplementary Information \cite{supplementary}). 
We have also applied our methods in collaboration with biologists to study various biological systems, ranging from innate immune pathways \cite{Liu2016} to cell death/survival pathways \cite{Liu2014,Liu2017,kagan2017}. Here we report on a selected list of applications of our SMC based analysis framework. We begin with models based on a single system of ODEs in Sect. \ref{sec:ode_apps}
and then describe applications to hybrid systems in Sect. \ref{sec:hybrid_apps} with further case studies and details in \cite{supplementary}.

\vspace{-0.5cm}
\subsection{SMC based analysis of ODE models}
\label{sec:ode_apps}
\vspace{-0.3cm}

\paragraph{\textbf{The MLC phosphorylation pathway}}
\vspace{-0.2cm}
We applied our SMC based parameter estimation procedure to the myosin light chain (MLC) phosphorylation pathway model \cite{mlc}. MLC phosphorylation is crucial to smooth muscle contraction. The sustained phosphorylation of MLC is associated with abnormal smooth muscle contraction, which can cause various diseases including hypertension and vasospasm \cite{mlc2}. The detailed reaction scheme of the thrombin-induced MLC phosphorylation pathway model can be found in  \cite{supplementary}. 

The model consists of $105$ ODEs. We assumed that $100$ (out of $197$) parameters are unknown in order to illustrate the scalability of our methods. By taking into account the cell-cell variability, we sampled the assumed distributions of initial concentrations and simulated the population based experimental observations for $10$ species at $20$ time points ranging from $0$ to $1000$ seconds. We used the time course data of regulator of G protein signaling 2 (RGS2), Ca$^{2+}$, guanine nucleotide exchange factor for Rho (p115RhoGEF), Rho GTPase, the protein kinase C (PKC)-diacylglycerol (DAG) complex, myosin light chain 2 (MLC2), protein phosphatase 1 regulatory subunit 14A (CPI-17), and MLC phosphatase (MYPT1-PPase),  together with the dynamic trends of the activated thrombin receptor and  inositol trisphosphate receptor (IP3R) to estimate unknown parameters. We reserved the data regarding the MLC-Rho-kinase complex and the Rho-kinase-MYPT1 complex for the test data set. 

Parameter estimation took $48.8$ hours on a machine with an Intel Core i7 3.4Ghz processor and 8GB memory. As the subset of results shown in Fig.~\ref{Thrombin:train} demonstrate, model simulations with the estimated parameters match both the training (Fig.~\ref{Thrombin:train}(a)) and test data (Fig.~\ref{Thrombin:train}(b)).

\begin{figure}
\begin{minipage}[b]{0.4\linewidth}
\centering
\includegraphics[scale=0.45]{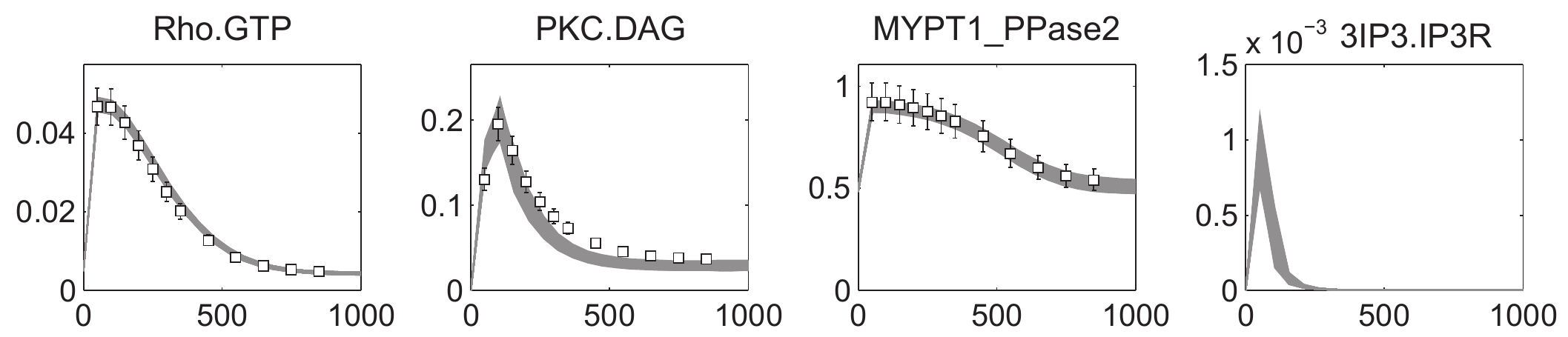}\\
\hspace{4cm} (a)
\end{minipage}
\hspace{4cm}
\begin{minipage}[b]{0.35\linewidth}
\centering
\includegraphics[scale=0.41]{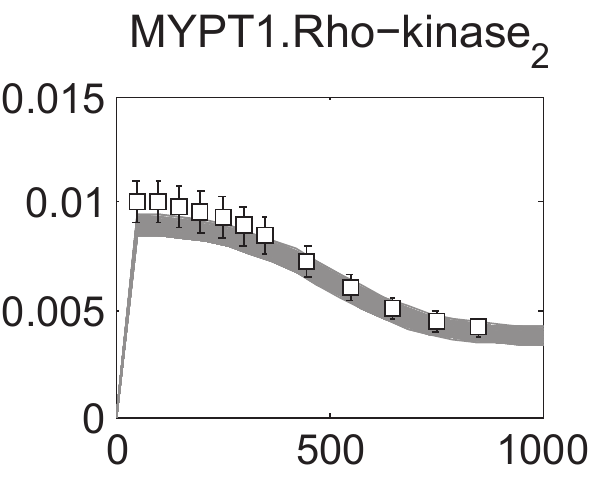}\\
(b)
\end{minipage}
\caption{Parameter estimation results. (a) Model simulations v.s. training data (b) Mode simulations v.s. test data. Reprinted by permission from \textit{Springer} \cite{palaniappan2013}.}
\vspace{-0.2cm}
\label{Thrombin:train}
\end{figure}

\paragraph{\textbf{The TLR3-TLR7 pathways crosstalk}}
\vspace{-0.2cm}
In collaboration with biologists, we studied the Toll-like receptor (TLR) pathways \cite{Liu2016}. TLRs are key to the innate immune system. They can recognize a broad range of microbial pathogens and stimulate the innate immune response to protect the host. It was previously believed that innate immune system has no ``memory''. However, we observed that macrophages can prime themselves after the first round of pathogen attack and induce a ``cytokine storm'' to combat  the subsequent intruding pathogens. Using a combined computational and experimental approach, we found that the short-term innate immunological memory in macrophages may be conferred by the TLR3 and TLR7 pathways crosstalk. 

\begin{figure}
\centering
\includegraphics[scale=0.65]{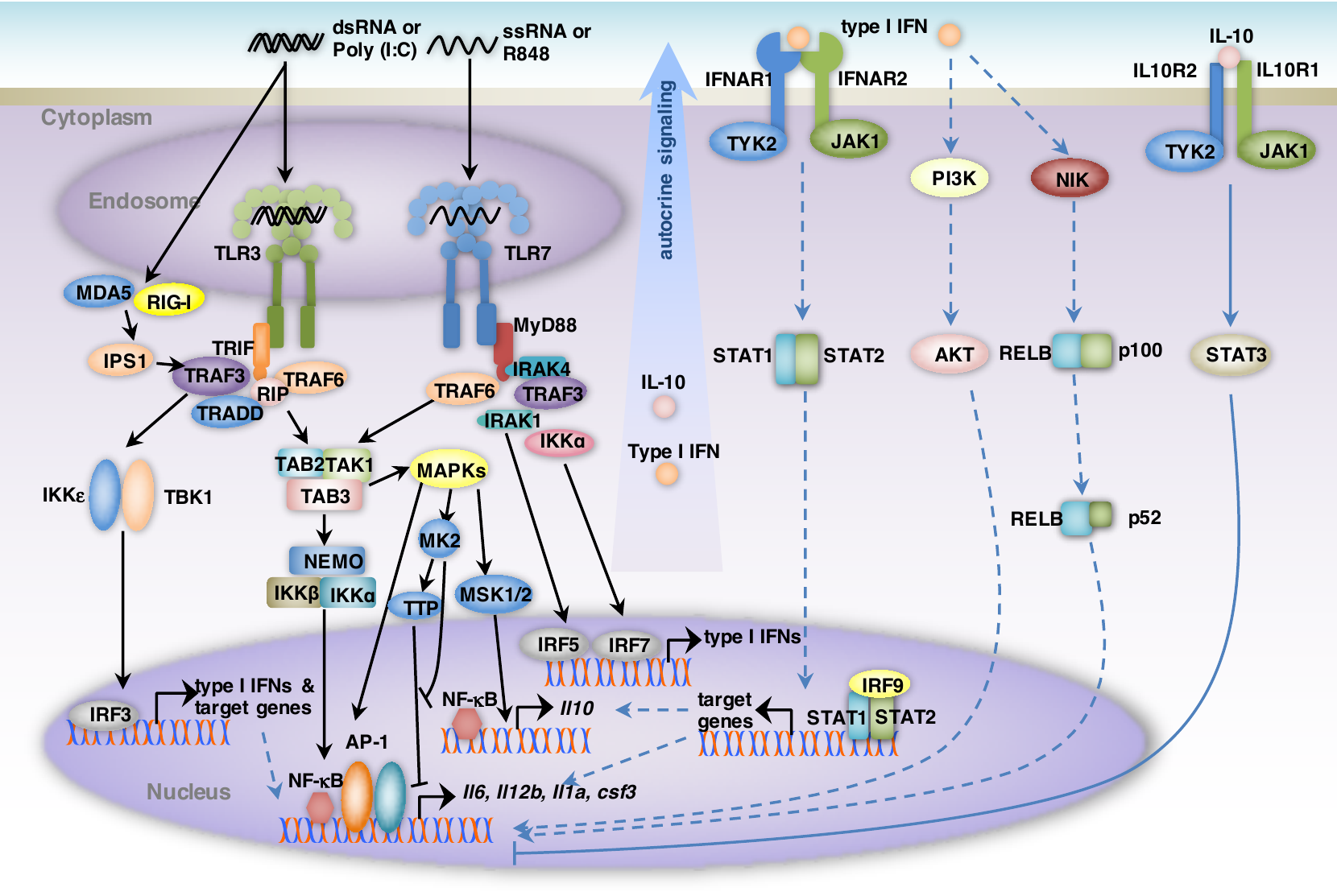}
\centering
\caption{The TLR3-TLR7 signaling network. Reprinted by permission from \textit{Science} \cite{Liu2016}.}
\label{tlr1}
\end{figure}

To identify the mechanism that is responsible for the synergistic production of cytokines, we built a kinetic model for the signal transduction network that underlies the activations of TLR3 and TLR7, which involves 112 ODEs and 129 unknown parameters (Fig. \ref{tlr1}). We carried out SMC based parameter estimation. To account for cellular heterogeneities, we allowed the initial concentrations to vary 5\% about their nominal values. The calibrated model reproduces the training data including the abundances of activated forms of p38, c-Jun N-terminal kinase (JNK), mitogen-activated protein kinase 1 (ERK), and NF$\kappa$B inhibitor $\alpha$ (I$\kappa$B$\alpha$) measured at 7 time points in response to single stimulation or two sequential stimulations with different time intervals between them (Fig. \ref{tlr2}a). The model predictions also match the test data including the mRNA levels of colony-stimulating factor 3 (\textit{Csf3}), interleukin 10 (\textit{Il10}), and Interleukin 1$\alpha$ (\textit{Il1a}) measured under various conditions (Fig. \ref{tlr2}b).

\begin{figure}
\centering
\includegraphics[scale=0.54]{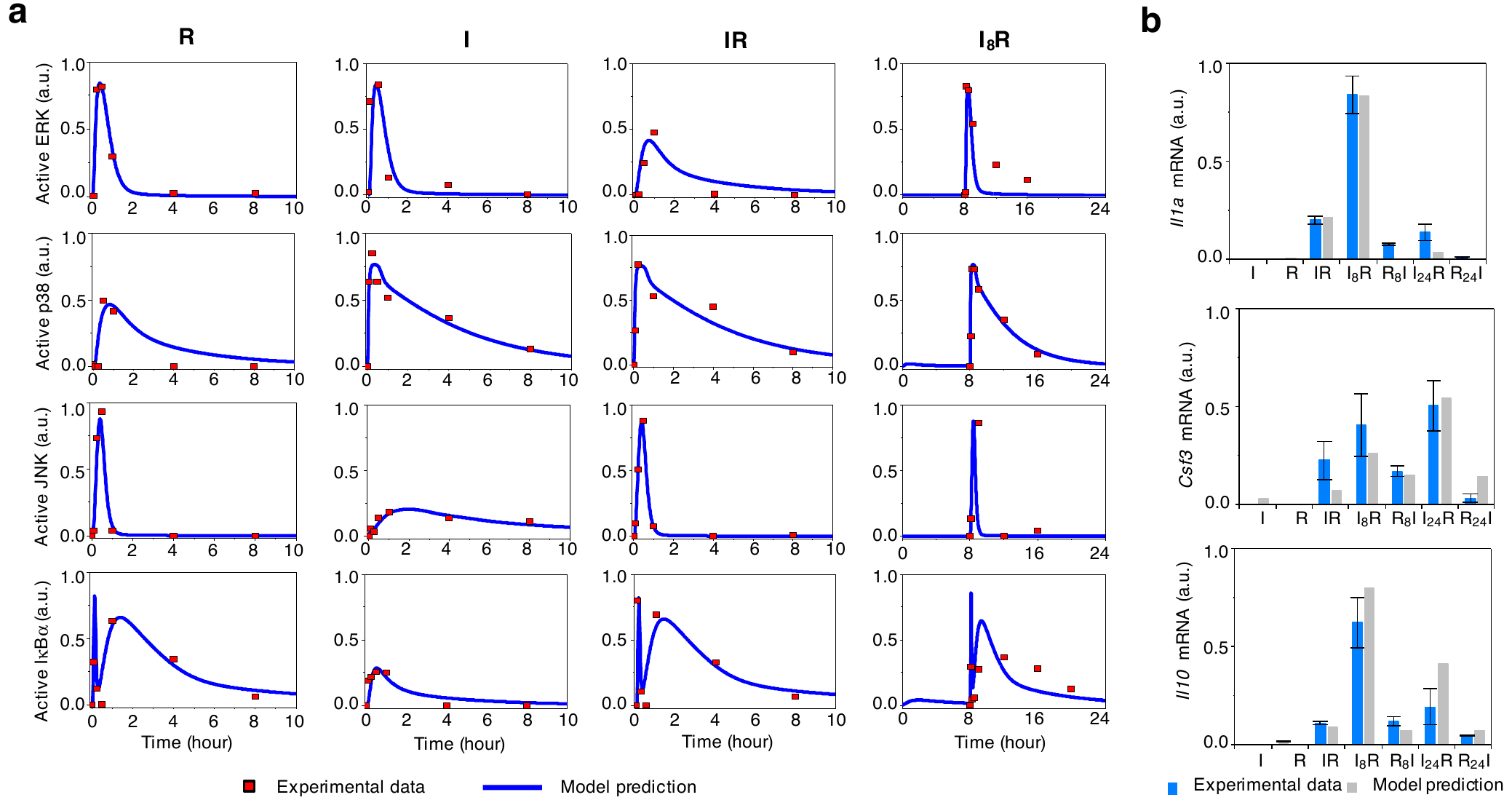}
\centering
\caption{Parameter estimation results. Model predictions vs. experimental (a) training data and (b) test data. Reprinted by permission from \textit{Science} \cite{Liu2016}.}
\label{tlr2}
\end{figure}

We next performed SMC based sensitivity analysis to identify critical species and reactions that govern the cytokine synergy. Our results suggested that the Janus kinase (JAK)-signal transducer and activator of transcription (STAT) pathway is responsible for synergistic production of inflammatory cytokines. Specifically, we found that STAT1 upregulates the production of inflammatory cytokines (e.g. interleukin 6 (IL-6) and interleukin 12 (IL-12)) by activating downstream transcription factors including interferon regulatory factor 1 (IRF1). On the other hand, STAT1 also promotes the production of IL-10, which in turn, suppress the production of inflammatory cytokines. Interestingly, the model reveals that the type 1 incoherent feedforward loop (I1-FFL) implicated by the opposing roles of STAT1 can induce a counterintuitive biphasic response: knocking down the initial concentration of STAT1 from its wild-type level can increase cytokine response, while the complete knockout of STAT1 markedly reduce cytokine response. Thus, STAT1 is a key regulator that not only boosts the immune response to kill intruding pathogens, but also protects host cells by maintaining the homeostasis. These model-driven findings were confirmed experimentally~\cite{Liu2016}.

\paragraph{\textbf{The Autophagy-apoptosis pathways crosstalk}}
\vspace{-0.2cm}
Autophagy is an essential cytoprotective mechanism that maintains the physiological function of the cell by removing outdated or damaged organelles and harmful aggregates of misfolded proteins, and protects cells against various stresses such as starvation and infections. Apoptosis is a common form of programmed cell death. Autophagy and apoptosis are implicated in the pathogenesis of many diseases. Both processes can be induced by the same type of stress signals and key mediators are shared between the two processes. However, it is far from known how cells evaluate the level of the stress signals and make a decision on their survival or death. We built a comprehensive model for the cell decision between autophagy and apoptosis \cite{Liu2017}, which covers the major signal transductions in response to starvation, DNA damage, and endoplasmic reticulum (ER) stress (Fig. \ref{aa1}). Our model involves 94 ODEs and 124 unknown parameters.   

\begin{figure}
\centering
\includegraphics[scale=0.45]{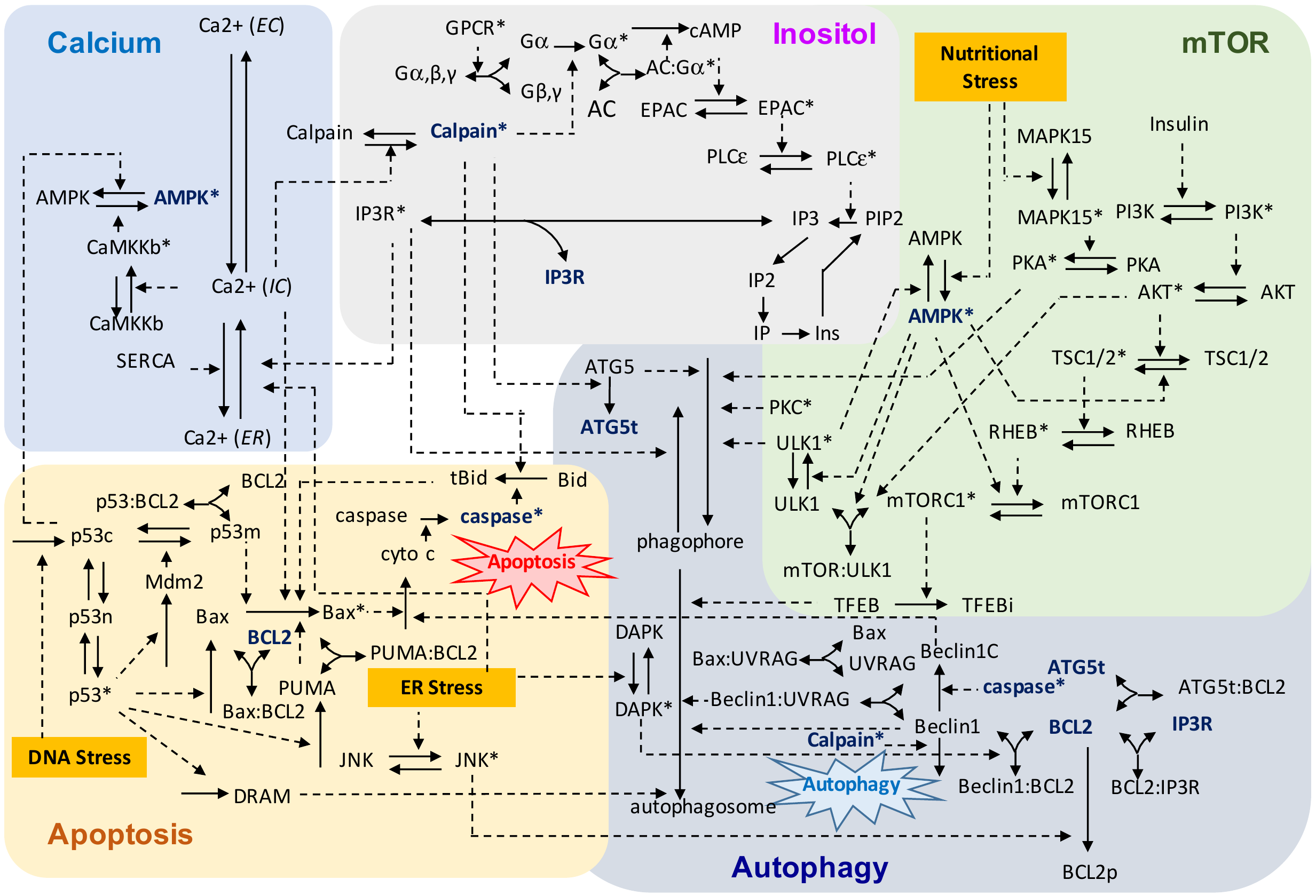}
\centering
\caption{The reaction scheme of autophagy-apoptosis signaling network. Reprinted by permission from \textit{Nature} \cite{Liu2017}.}
\vspace{-0.3cm}
\label{aa1}
\end{figure}

We performed parameter estimation using our SMC method. The training data included the time profiles of apoptotic and autophagic levels of H4 cells in response to different doses of  torin 1, staurosporine, and tunicamycin treatments (Fig. \ref{aa2}a). The test data included time profiles of apoptotic and autophagic levels of RPTC cells in response to cisplatin. Fig.~\ref{aa2} shows that simulation results match both the training and test data.

\begin{figure}
\centering
\includegraphics[scale=0.6]{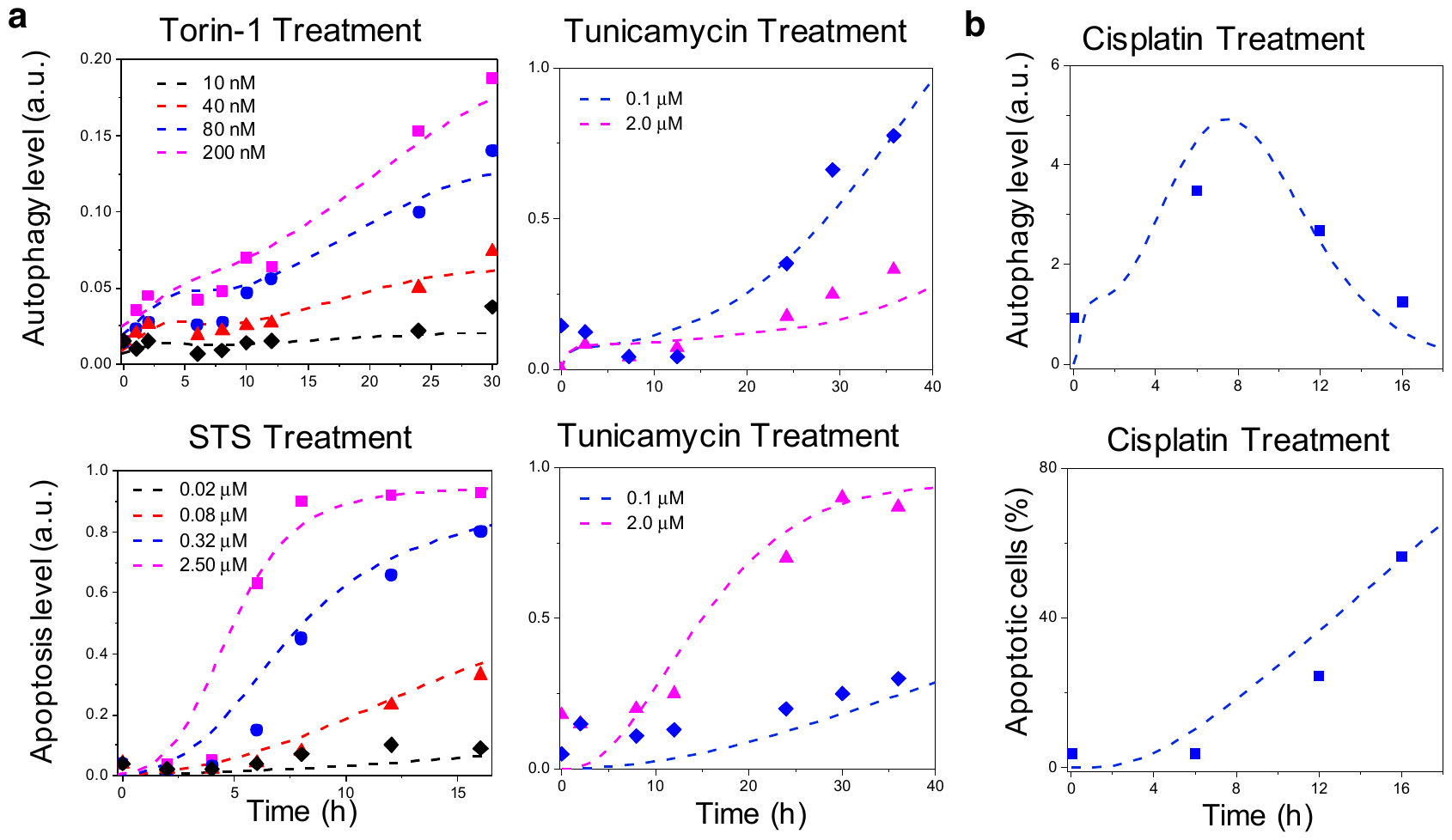}
\centering
\caption{Parameter estimation results. (a) Model simulations v.s. training data. (b) Model simulations v.s. test data.  Reprinted by permission from \textit{Nature} \cite{Liu2017}.}
\vspace{-0.5cm}
\label{aa2}
\end{figure}

We next performed SMC based sensitivity analysis. We found that p53, cytoplasmic Ca$^{2+}$, AMP-activated protein kinase (AMPK), and calpain are the most important regulators of cell fate decision. They interconnect and form a core regulatory network, which tightly controls the decisions between autophagy and apoptosis through multiple feedforward and feedback loops. Specifically, the model reveals that cytoplasmic Ca$^{2+}$ level is maintained by a feedback loop mediated by IP3 signaling. Cytoplasmic Ca$^{2+}$ regulates apoptotic and autophagic responses mainly through an incoherent feedforward loop mediated by AMPK. Interestingly, increased cytoplasmic Ca$^{2+}$ can either inhibit or enhance autophagy. The overall role of cytoplasmic Ca$^{2+}$ in modulating autophagy is determined by the expression level of Ca$^{2+}$/calmodulin-dependent protein kinase kinase-beta (CaMKK$\beta$). This partly explains the varied efficacy of Ca$^{2+}$-modulating drugs in regulating cell fate observed in different tissues.

\vspace{-0.5cm}
\subsection{SMC based analysis of hybrid automata based models}
\label{sec:hybrid_apps}
\vspace{-0.3cm}
We have also applied our method to hybrid automaton models of the cardiac cell \cite{fenton08} and  the circadian rhythm network \cite{miyano06}. Here we focus on the former.

\vspace{-0.5cm}
\paragraph{\textbf{The cardiac cell hybrid model}}
\vspace{-0.2cm}
Ion channels on the membrane of cardiac cells regulate their electrodynamics  and control the heart rhythm. Disordered electrical activation of cardiac cells is associated with heart diseases including \textit{fibrillation} and \textit{tachycardia}. Fig.~\ref{heart} shows a hybrid automaton model of the cardiac electrodynamics of a ventricular cell. The parameter values for each of the three subtypes of ventricular cell: endocardial, midmyocardial, and epicardial cells can be found in \cite{supplementary}.

The model has four state variables. $v$, $w$ and $s$ denotes channel gates, while $u$ denotes the potential across the cell membrane. $\epsilon$ denotes the external stimulus. The cardiac action potential (AP) is a change in $u$ in response to $\epsilon$.

\begin{figure}[ht]
	\centering
	\includegraphics[width=1\textwidth]{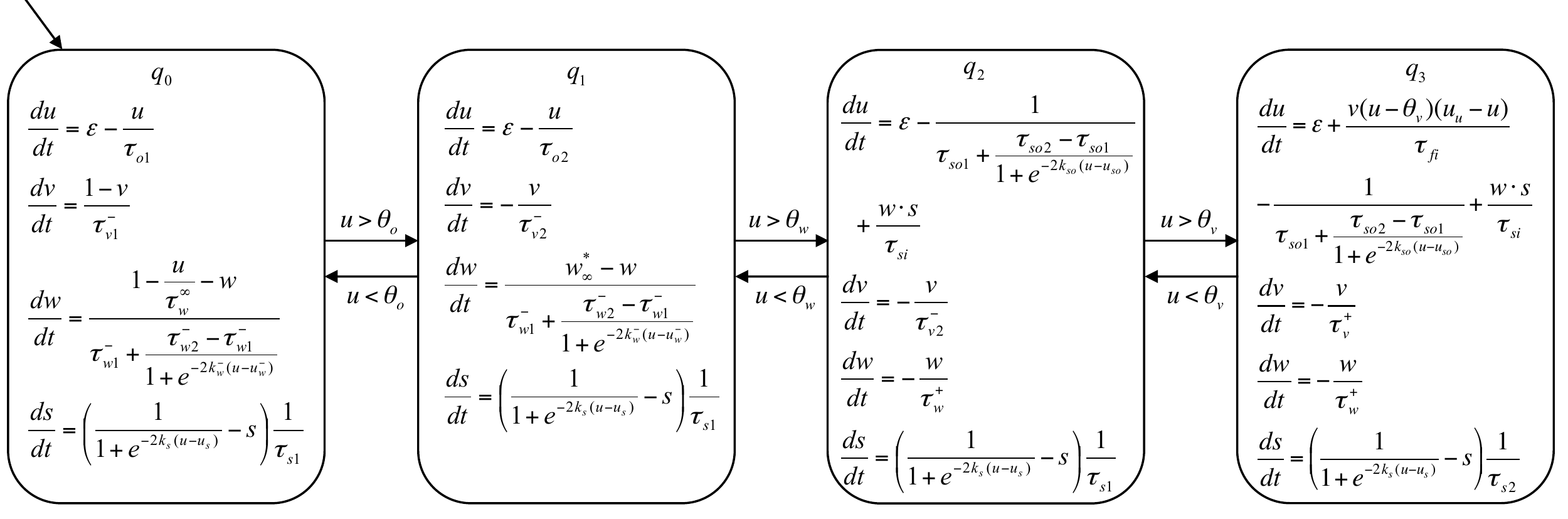}
	\caption{The cardiac cell hybrid model \cite{grosu11}. Reprinted by permission from \textit{Springer} \cite{gyori2015}.}
	\label{heart}
\vspace{-0.3cm}
\end{figure}

The model has four modes. In mode $q_0$, $s$ is closed and $v$ and $w$ are open. The cell is resting. Given a stimulation of  $\epsilon =  1$mV, which lasts for $1$ millisecond, the system can switch to mode $q_1$ when the increased $u$ reaches the threshold $\theta_{o}$. $v$ starts closing in mode $q_1$, which alters the decay rate of $u$. When $u > \theta_w$, the system will switch to mode $q_2$, in which  $w$ starts closing. The system will finally reach mode $q_3$ when $u > \theta_v$, which completes the process of an ``AP initiation''. $u$ will reach its peak in mode $q_3$, and then starts decreasing, which will induce a contraction in heart muscle.

\runinhead{Property 1}
Insufficient excitability of the cardiac cell can cause ventricular tachycardia and fibrillation. The property below specifies that the system can leave mode $q_0$ in response to a stimulation:
\begin{center}
	$\mathbf{F}^{\leq 500}(\neg[q_0])$.
\end{center}
The verification results show that under the normal condition, this property holds for midmyocardial, endocardial, and epicardial cells in response to $\epsilon =  1$mV. Under a disease condition (e.g.  $\tau_{o2} = 0.1$ or $\tau_{o1} = 0.004$, see \cite{delta}) the property does not hold for any $\epsilon$. This implies that the impulse conduction can be blocked by a region of unexcitable cells, which can cause ventricular disorders. This is consistent with \cite{tanaka07}.

\runinhead{Property 2} After reaching the mode $q_3$ (i.e. a successful AP was completed, the cardiac cell should return to mode $q_0$ and wait for the next stimulation. We specify this property as follows:
\begin{center}
	$\mathbf{F}^{\leq 500}([q_3] \wedge \mathbf{F}^{\leq 500}(\mathbf{G}^{\leq 100}([q_0])))$.
\end{center}
The verification results show this property holds for endocardial, epicardial, midmyocardial cells under the normal condition in response to a transient stimulation. If we set the duration of $\epsilon$ to be $500$ milliseconds (i.e. sustained stimulation), the property does not hold. This implies that cardiac cells are unable to reach the resting state. It is known that the stimulus $\epsilon$ can be generated by neighboring cells in ventricular tissue \cite{fenton08}. Our results suggest that the stimulation profile of the neighboring cells can shape the stimulation profile of a single cell. 

\begin{figure}[ht]
	\centering
	\includegraphics[width=0.9\textwidth]{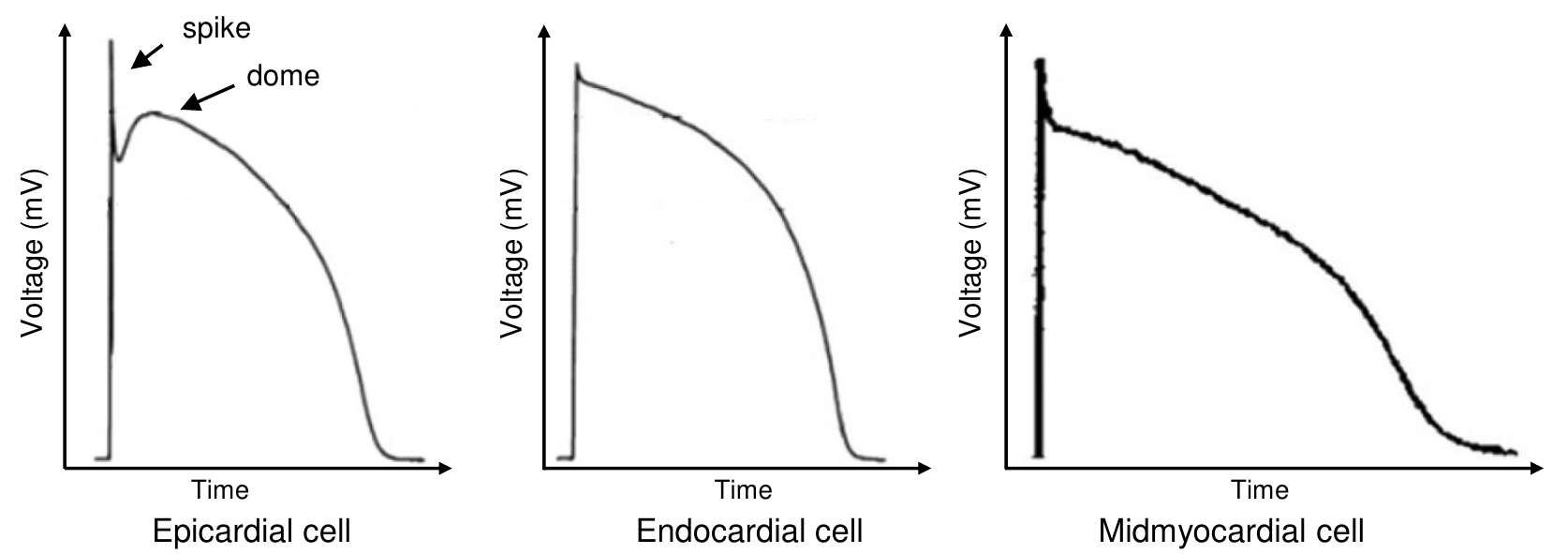}
	\caption{The morphologies of AP in endocardial \cite{nabauer96}, epicardial \cite{nabauer96}, and midmyocardial \cite{drouin95} cells. Reprinted by permission from \textit{Springer} \cite{gyori2015}.}
	\label{morphology}
\vspace{-0.3cm}
\end{figure}

\runinhead{Property 3} Midmyocardial, epicardial, and endocardial cells have different morphologies of AP profile \cite{nabauer96,drouin95}. Specifically, a critical ``spike-and-dome'' (see Fig. \ref{morphology}, \textit{left panel})
AP morphology can not be observed in midmyocardial and endocardial cells but can be observed in epicardial cells (Fig.~\ref{morphology}).
The property below specifies that the system has a spike-and-dome AP morphology.
\begin{center}
	$\mathbf{F}^{\leq 500}(\mathbf{G}^{\leq 1}([1.4 \leq u]) \wedge \mathbf{F}^{\leq 500}([0.8 \leq u] \wedge [u \leq 1.1] \wedge \mathbf{F}^{\leq 500}(\mathbf{G}^{\leq 50}([1.1 \leq u]))))$.
\end{center}
The verification results show that under the normal condition this property does not hold for midmyocardial and endocardial cells  and \emph{true} for epicardial cells  in response to a transient stimulation. By perturbing each model parameter and checking if the property still holds in epicardial cells, we found that  $\tau_{s2}$ is critical to the spike-and-dome morphology of AP. The property does not hold when $\tau_{s2}=2$, suggesting that gate $s$ is crucial to epicardial cells. This is supported by the results reported in \cite{delta} that a model without $s$ \cite{fenton98} is unable to reproduce the experimentally observed electrodynamics in epicardial cells.

\vspace{-0.5cm}
\section{Discussion}
\label{sec:conclusion}
\vspace{-0.3cm}
In this chapter, we presented approaches to analyzing biopathway models with the help of statistical model checking via a probabilistic transformation  of the behaviors of  deterministic ODEs, and non-deterministic hybrid automata. The main motivation is to handle cell-cell variability in the dynamics that arises due to fluctuations in the initial concentrations and parameter values across a cell population.  We obtain the required transformations of the dynamics in a mathematically grounded way by imposing reasonable continuity assumptions and exploiting basic aspects of measure theory and the theory of ODEs. The additional ingredient required from the modeling side is intervals of values for the initial conditions and kinetic rate constants as well as a probability distribution over this set of values. Our transformations opens the path to applying scalable statistical analysis methods to complex dynamics in a variety of settings.

To bring this out, we have constructed a SMC based approach for parameter estimation. In this setting BLTL is used to represent both quantitative time course experimental data and known qualitative dynamical properties of pathways. The probability of the BLTL specification being satisfied by the dynamics according to the chosen parameter values is assessed by the SMC procedure. By equipping existing parameter space exploring strategies with this pointwise evaluation method, one can arrive both novel and efficient parameter estimation and sensitivity analysis methods.

We have used the MLC phosphorylation pathway model to illustrate the applicability of our parameter estimation method. We have also applied our method to the studies of Toll-like receptor \cite{Liu2016}, autophagy \cite{Liu2017}, p53 \cite{Liu2014}, and ferroptosis \cite{kagan2017} pathways. In these applications our method has lead to estimated parameters with good quality given noisy cell-population based experimental observations. These results also demonstrated the scalability of our method in practical settings. It is worth pointing out that our method is a generic one. For instance it has been generalized to apply to rule-based models \cite{Liu2016bibm}. The scalability and computational efficiency of this parameter estimation technique has been further enhanced by a GPU-based implementation \cite{Ramanathan2015}.

We also presented a probabilistic verification method for analyzing the rich dynamics of a hybrid automaton $H$ using an associated  Markov chain $M$. We have established a strong relationship between the behaviors of $H$ and $M$ for time bounded properties. As a result, the intractable verification problem for $H$ can be solved approximately using $M$. Concretely, we have constructed an SMC procedure to verify that $M$ almost certainly meets a BLTL specification. The two case studies we have carried out demonstrate the applicability of this novel analysis method for hybrid systems. 

Turning to future work, our approach has the potential to be extended to a network of hybrid systems. In such settings one can model crosstalk, feedback and feedforward loops that involve multiple signaling pathways. It will also be fruitful to develop GPU based methods to accelerate the trajectory sampling procedure of hybrid models as also to explore parameter landscapes in parallel to identify likely regions that will lead to the desired system responses to stimuli of interest. We plan to integrate our methods into a generic SMC based tool to enable automatic model calibration, validation, and analysis. The tool will be equipped with interactive web-based user interface to enable non-experts to convert different types of experimental data in to BLTL formulas. An interesting challenge will be to develop NLP approaches to enable the automatic extraction of qualitative biological knowledge from existing databases and literature.  In this context, literature-driven pathway model construction tools such as INDRA \cite{gyori2017word} promise to offer significant help.

\section*{Acknowledgements} 
This work was partly supported by the NIH awards P41GM103712, U19AI068021 and P30DA035778.

\bibliographystyle{spmpsci}
\bibliography{ref}
\end{document}